\documentclass{article}
\newcommand{\efx}{\textsf{EFX}}
\newcommand{\ef}{\textsf{EF}}
\newcommand{\np}{\textsf{NP}}
\newcommand{\mms}{\textsf{MMS}}
\newcommand{\po}{\textsf{PO}}

\usepackage{graphicx} % Required for inserting images
\usepackage[letterpaper, left=1in, right=1in, top=0.9in, bottom=0.9in]{geometry}
\usepackage[american]{babel}
\usepackage[normalem]{ulem}
\usepackage{amsmath, amssymb, cases, amsthm}
\usepackage{mdframed,lipsum}
\usepackage{thmtools, thm-restate}
\usepackage{tikz}
\usepackage[shortlabels]{enumitem}
\usepackage{mdframed}
\usepackage{bbm}
\usepackage{bm}
\usepackage{microtype}
\usepackage{xcolor}
\usepackage{makecell}
\usepackage{mathtools}
\usepackage{algorithmic}
\usepackage[procnumbered,ruled,vlined,linesnumbered]{algorithm2e}
\usepackage{float}
\usepackage{varwidth}
\usepackage{svg}
\usepackage{tikz-network}
\usetikzlibrary{decorations.pathreplacing}
\usepackage{subcaption} % For subfigures and labeling (a), (b)
\usepackage{multirow}
\usepackage{natbib} 
\usepackage{tablefootnote}
\usepackage{thm-restate}

\RequirePackage{etex} 

\SetKwInput{KwData}{Input}
\SetKwInput{KwResult}{Output}
\SetKwInput{KwGlobalVar}{Global variables}

\usepackage[bookmarks,colorlinks,breaklinks]{hyperref}
\hypersetup{urlcolor=blue, colorlinks=true, citecolor=green!50!black, linkcolor=blue}
\usepackage[capitalize,nosort,nameinlink]{cleveref}
\usepackage{etex}
\DeclareMathOperator*{\argmax}{arg\,max}

\newcommand{\pluseq}{\mathbin{\mbox{+=}}}

\declaretheorem[numberwithin=section,refname={Theorem,Theorems},Refname={Theorem,Theorems}]{theorem}

\declaretheorem[numberlike=theorem]{lemma}
\declaretheorem[numberlike=theorem]{proposition}
\declaretheorem[numberlike=theorem]{corollary}
\declaretheorem[numberlike=theorem]{definition}

\declaretheorem[numberlike=theorem,style=remark]{remark}

\declaretheorem[numberlike=theorem, refname={Observation,Observations},Refname={Observation,Observations},name={Observation}]{observation}

\theoremstyle{definition}
\title{EFX Allocations and Orientations on Bipartite Multi-graphs: \\ A Complete Picture}

\date{}

\author{
  Mahyar Afshinmehr\thanks{A part of the work was done when MA was an intern at MPI-INF, Germany.}\\
   \texttt{mahyarafshinmehr@gmail.com} \\
   \and
   Alireza Danaei \\
   Massachusetts Institute of Technology, USA\\
   \texttt{alirez11@mit.edu}
   \and 
   Mehrafarin Kazemi\thanks{A part of the work was done when MK was an intern at MPI-INF, Germany.} \\
   Max Planck Institute for Software Systems, SIC, Germany\\
   \texttt{mkazemi@mpi-sws.org}
   \and
   Kurt Mehlhorn\\
   Max Planck Institute for Informatics, SIC, Germany\\
   \texttt{mehlhorn@mpi-inf.mpg.de}
\and
   Nidhi Rathi\\
   Max Planck Institute for Informatics, SIC, Germany\\
   \texttt{nidhirathi@gmail.com}
}

\begin{document}

\maketitle

\begin{abstract}
    We consider the fundamental problem of \emph{fairly} allocating a set of indivisible items among agents having valuations that are represented by a \emph{multi-graph} -- here, agents appear as vertices and items as edges between them and each vertex (agent) only values the set of its incident edges (items). The goal is to find a fair, i.e., \emph{envy-free up to any item} ($\efx$) allocation. This model has recently been introduced by \cite{christodoulou2023fair} where they show that $\efx$ allocations always exist on simple graphs for \emph{monotone} valuations, i.e., where any two agents can share at most one edge (item). A natural question arises as to what happens when we go beyond simple graphs and study various classes of multi-graphs?  
    
    We answer the above question affirmatively for the valuation class of \emph{bipartite multi-graphs} and \emph{multi-cycles}. The main contribution of this work is to establish the existence of $\efx$ allocations on bipartite multi-graphs for \emph{monotone} valuations and on multi-cycles for \emph{$\mms$-feasible} valuations. We also present pseudo-polynomial time algorithms to compute $\efx$ allocations for the above settings. Furthermore, we show that for bipartite multi-graphs with \emph{cancelable} valuations, $\efx$ allocations can be computed in polynomial time. We thus deepen the understanding of $\ef$ allocations by expanding the spectrum of settings in which they are guaranteed to exist for an arbitrary number of agents.

    Next, we study $\efx$ \emph{orientations} (allocations where every item is assinged to one of its two endpoint agents) and provide a complete characterization of their existence on bipartite multi-graphs in terms of two key parameters---the number of edges shared between any two agents and the diameter of the graph. Finally, we prove that it is $\np$-complete to determine whether a given fair division instance on a bipartite multi-graph admits an $\efx$ orientation, even with a constant number of agents.

\end{abstract}
\section{Introduction}

The theory of \emph{Fair Division} formalizes the classic problem of dividing a collection of resources to a set of participating players (referred to as \emph{agents}) in a \emph{fair} manner. This problem forms a key concern in the design of many social institutions and arises naturally in many real-world scenarios such as dividing business assets, assigning computational resources in a cloud computing environment, air traffic management, course assignments, divorce settlements, and so on (\cite{etkin2007spectrum,moulin2004fair,vossen2002fair,budish2012multi,pratt1990fair}). The theory of fair division lies at the crossroads of economics, social science, mathematics, and computer science, with its formal exploration beginning in 1948 by \citeauthor{steinhaus1948problem}, Banach, and Knaster. In recent years, this field has experienced a flourishing flow of research; see \cite{survey2023,brams1996fair,brandt2016handbook,robertson1998cake} for excellent expositions.

In this work, we study the \emph{discrete} setting of fair division where we wish to divide a set of $m$ indivisible items to a set $\{1,2, \dots, n\}$ of $n$ agents with each item being allocated \emph{wholly} to some agent. The standard notion of fairness in this line of work is that of \emph{envy-freeness}, which entails a division (allocation) as \emph{fair} if every agent values her bundle as much as any other bundle in the allocation (\cite{foley1967resource}). Since envy-free allocations may not always exist for the case of indivisible items,\footnote{Consider an instance with two agents valuing a single item positively. Here, the agent who does not receive the item will envy the other. In contrast, envy-free allocations always exist when the resource to be allocated is divisible (see \cite{stromquist1980cut,edward1999rental}).} several variants of envy-freeness have been explored in the literature. Among all, \emph{envy-freeness up to any item} ($\efx$) is considered to be the flag-bearer of fairness for the discrete setting (introduced by \cite{caragiannis2016unreasonable}). We say an allocation $X = (X_1, X_2, \dots, X_n)$ with bundle $X_i$ being allocated to agent $i \in [n]$ is $\efx$ if for every pair $i,j \in [n]$ of agents, agent $i$ prefers her bundle $X_i$ over $X_j \setminus \{g\}$ for \emph{every} item $g \in X_j$. $\efx$ is considered to be the ``closest analogue of envy-freeness'' for the discrete setting (\cite{caragiannis2019envy}). 

%A weaker notion is that of \emph{envy-freeness up to some item} ($\ef1$) \cite{} wherein the envy for any agent must go away after removing some item from the envied bundle in the allocation. It is known that $\ef1$ allocations always exist and can be computed efficiently as well \cite{}. Clearly, any $\efx$ allocation is also $\ef1$.

It remains a major open problem to understand the existence of $\efx$ allocations (\cite{procaccia2020technical}). Despite significant efforts, it is not known whether $\efx$ allocations exist for four or more agents, even with additive valuations (\cite{chaudhury2020efx}). This naturally points towards the complexity of this problem and motivates the study of its various kinds of relaxations. That is, one may begin to understand the concept of $\efx$ either by finding approximate $\efx$ allocations, or by studying special valuation classes, or by relaxing the notion of $\efx$ via \emph{charity} or its \emph{epistemic} form. We refer the readers to Section~\ref{sec:related} for further discussion. 

A question of interest is to understand for what settings or valuation classes, $\efx$ allocations are guaranteed to exist. To this effect, in this work, we consider a recently introduced model by \cite{christodoulou2023fair} where agent valuations are represented via a graph (or a multi-graph). Here, each item can have a positive value for at most two agents, and such an item is represented by an edge connecting those two agents. Equivalently, items are valued at zero by all agents who are not their endpoints. Thus, the vertices of the multi-graph correspond to agents, and the multi-edges between any pair of agents represent the items valued exclusively by them. For each agent $i \in [n]$, we assume a valuation function $v_i$ such that an edge $g$ has a positive marginal value ($v_i(S \cup g) - v_i(S) > 0$ for every subset $S$) if and only if $g$ is incident to $i$. We call this a \emph{fair division problem on multi-graphs} with the goal of finding complete $\efx$ allocations.  

This model captures situations in which the number of agents interested in any given item is limited—specifically, when an item is \emph{relevant} or \emph{valuable} to at most two agents. The graph-based formulation is inspired by geographic contexts where agents care only about nearby resources. For example, a trade corridor or a natural gas pipeline benefits only the bordering countries involved, while others gain nothing. Likewise, in the New Champions League format, each football match takes place at the home ground of one of the two participating teams, and the remaining teams are unaffected.

\cite{christodoulou2023fair} proved that $\efx$ allocations exist for all simple graphs, where any pair of agents vertices share at most \emph{one} edge item. This naturally raises the question of what happens in more general classes of multi-graphs, where a pair of agents can have multiple valuable items in common?

\subsection{Our Contribution and Techniques}
In this work, we answer the above question raised by \cite{christodoulou2023fair} and study fair division instances where agent valuations are represented via \emph{bipartite multi-graphs} and provide a complete picture of $\efx$ allocations for such instances. Note that bipartite multi-graphs are a huge class consisting of multi-trees, multi-cycles with even length, and planar multi-graphs with faces of even length, to name a few. 

We also study special types of (\emph{non-wasteful}) allocations---$\efx$ \emph{orientations}---where every item is allocated to one of the agent endpoints. Note that $\efx$ orientations are desirable since every item is allocated to an agent who values it, and hence, is \emph{non-wasteful}. Our main results are as follows.

\begin{table}[h!]
\centering
{
\begin{tabular}{|c|c|}
\hline
Parameters & $\efx$ Orientation \\
\hline
$\textit{acyclic}, q = 2$, $ d(G) \leq 4$ &  exists (Theorem \ref{thm:orientation-P5}) \\
\hline
$\textit{acyclic}, q = 2$,  $d(G) > 4$  & may not exist (Theorem \ref{thm:orientation-counter-example-q2})\\
\hline
$\textit{acyclic}, q > 2$, $  d(G)\leq 2$ & exists (Proposition \ref{prop:multi-stars})\\
\hline
$\textit{acyclic}, q > 2$,  $d(G) > 2$ & may not exist (Theorem \ref{thm:orientation-counter-example-bigq})\\
\hline
$\textit{cyclic}, q \geq 2$,  $ d(G) \geq 2$  & may not exist (Theorem \ref{thm:orientation-counter-example-cycle})\\
\hline
\end{tabular}
}

\caption{A complete picture for $\efx$ orientations on bipartite multi-graphs with additive valuations based on $q$ and $d(G)$.} \label{tab:orientation}
\end{table}

\begin{itemize}
    \item The primary contribution of this work is to establish the existence of $\efx$ allocations for fair division instances on \emph{bipartite multi-graphs} with general \emph{monotone} valuations (Theorem~\ref{thm:multi-bipartite-main}). Moreover, we can compute an $\efx$ allocation in pseudo-polynomial time for monotone valuations and in polynomial time for \emph{cancelable} valuations (Theorem~\ref{thm:multi-bipartite-main}). It is relevant to note that \emph{cancelable} valuations are a strict superclass of additive valuations (see Section~\ref{sec:prelim} for definitions).

    Hence, we answer the open problem listed in \cite{christodoulou2023fair} and deepen our understanding of the existence of $\efx$ allocations.

    \item  While extending our techniques beyond bipartite multi-graphs, we also prove that $\efx$ allocations exist and can be computed in pseudo-polynomial time on \emph{multi-cycle graphs} with $\mms$-\emph{feasible} valuations (Theorem~\ref{thm:multi-cycle}). Note that $\mms$-\emph{feasible} valuations are a strict superclass of cancelable valuations.

    We also show that for multi-cycles with at least four agents, $\efx$ allocations always exist and can be computed in pseudo-polynomial time for monotone valuations and in polynomial time for cancelable valuations (Corollary \ref{cor:cycle}).
    
    \item Next, we show that $\efx$ \emph{orientations} may not always exist for fair division instances on bipartite multi-graphs; in particular, they may not exist even in very simple settings of multi-trees. Nonetheless, we provide a complete characterization of the existence of $\efx$ orientations for monotone valuations depending on two key parameters: $q$, the maximum number of edges shared between any two agents, and $d(G)$, diameter of the graph $G$ (see Table~\ref{tab:orientation}). 

    The fact that $\efx$ orientations do not always exist can be seen as a proof to show that inefficiency is inherent and hence approximations are necessary. We show that, for additive valuations, there exist orientations on bipartite multi-graphs where at least $\lceil\frac{n}{2}\rceil$ agents are $\efx$ and the remaining agents are $1/2$-$\efx$ (Theorem~\ref{thm:1/2-efx}).

    \item We also show that for monotone valuations, we can compute $\efx$ orientations in polynomial time when the diameter of the acyclic bipartite multi-graph is at most four and any two adjacent vertices share at most two edge items (Theorem~\ref{thm:orientation-P5}). 
    
    \item It is $\np$-complete to decide whether a given fair division instance on bipartite multi-graphs admits an $\efx$ orientation, even with a constant number of agents (Theorem~\ref{thm:hardness-orientation-multi-tree}).
\end{itemize}

Two other papers independently and in parallel showed the existence of $\efx$ allocations involving multi-graphs. \cite{SS25} independently show that the $\efx$ allocations always exist in bipartite multi-graphs for monotone valuations. This coincides with our first result, but we further prove polynomial-time computability of $\efx$ allocations (in the above setting) for cancelable valuations and in pseudo-polynomial time for monotone valuations. Moreover, \cite{SS25} proves the existence of $\efx$ allocations for multi-graphs when each agent has at most $\lceil n/4 \rceil -1$ neighbors, where $n$ is the total number of agents, or when the shortest cycle with non-parallel edges has length at least $6$. 

\cite{BP24} independently showed the existence
of EFX allocations in bipartite multi-graphs for cancelable valuations and in multi-trees for general monotone valuations. Our first result generalizes these two results. Moreover, they showed that multi-graphs with chromatic number $t$ admit EFX allocations when the shortest cycle using non-parallel edges is of length at least $2t-1$; this result holds when agents have cancelable valuations. Note that the multi-graphs with chromatic number 2 are just the bipartite multi-graphs.\\

\noindent
\textbf{Technical Overview:}
 We will give a description of the main techniques developed in this work to prove the existence of $\efx$ allocations on bipartite multi-graphs with monotone valuations. For simplicity, we will assume additive valuations here. For a given bipartite multi-graph $G = (S \sqcup T, E)$, let us consider two adjacent vertices $i \in S$ and $j \in T$. The starting point in our technique is inspired by the \emph{cut-and-choose} protocol that is used to prove the existence of $\efx$ allocations for two agents.\footnote{For two agents, the first agent divides the set of items into two $\efx$-feasible bundles (for her), and the second agent chooses her favorite bundle of the two, while the other bundle goes to the first agent.} Based on this protocol, we will define one specific \emph{configuration} for the set of items, $E(i,j)$, between $i$ and $j$. In this configuration, $j \in T$ will cut the set $E(i,j)$ into two bundles $C_1$ and $C_2$ such that she is $\efx$-happy with both bundles. We call this the \emph{$j$-cut configuration} between agents $i$ and $j$. Note that such a configuration can be computed in polynomial time for additive valuations.

The idea is to find a partial $\efx$ orientation $X$ that satisfies a set of six useful properties (listed in Section~\ref{sec:definition-4}). These properties pave a simple way for us to extend $X$ to a complete allocation while maintaining $\efx$ guarantees. One of them ensures that for any $i \in S$, $j\in T$, either (i) no item from $E(i, j)$ is allocated, (ii) exactly one of $C_1$ or $C_2$ is allocated to either $i$ or $j$, or (iii) both $C_1$ and $C_2$ are allocated to $i$ and $j$, such that one receives $C_1$ and other $C_2$ in $X$.

 To create such a partial orientation, we start with a greedy algorithm that allocates a set of items to each agent such that every agent in $T$ is non-envied. This initial step ensures that the set $X_i$ of items allocated to any agent $i \in [n]$ is such that $X_i \subseteq E(i,j)$ for some $j \in [n]$. Moreover, this partial orientation is $\efx$ wherein any vertex that has envy is certainly non-envied. We then try to orient unallocated items incident to a non-envied vertex to her while maintaining partial $\efx$ until the only unallocated edges are between a non-envied and an envied vertex. 

 Once we find a partial $\efx$ orientation with certain useful properties, we can go two ways (of our choice) to have a complete allocation. We can either compute, in polynomial time, (i) an orientation where at least $n/2$ agents are $\efx$, and the remaining agents are $1/2$-$\efx$ (Theorem~\ref{thm:1/2-efx}), or (ii) an exact $\efx$ allocation (Theorem~\ref{thm:multi-bipartite-main}). For the latter, we know that since $\efx$ orientations do not necessarily exist for bipartite multi-graphs, we have to allocate the remaining edges to a vertex other than their endpoints, which will create a wasteful (albeit an $\efx$) allocation.

Before finding a complete allocation, we have ensured, by one of our properties, that both non-envied and envied vertices are satisfied enough with what is allocated to them that they will not envy if we give all the unallocated edges adjacent to them to a specific third vertex. We, therefore, safely allocate the remaining items from the set $E(i,j)$ to a specific agent $k \neq i,j$, and finally compute an $\efx$ allocation.

\subsection{Further Related Work} \label{sec:related}

For the notion of $\efx$, \cite{plaut2020almost} proved its existence for two agents with monotone valuations. For three agents, a series of works proved the existence of $\efx$ allocations when agents have additive valuations (\cite{chaudhury2020efx}), \emph{nice-cancelable} valuations (\cite{berger2021almost}), and finally when two agents have monotone valuations and one has an \emph{$\mms$-feasible} valuation (\cite{akrami2025efx}). $\efx$ allocations exist when agents have identical (\cite{plaut2020almost}), binary (\cite{halpern2020fair}), or bi-valued (\cite{amanatidis2021maximum}) valuations.
The study of several approximations (\cite{chaudhury2021little,amanatidis2020multiple,chan2019maximin,farhadi2021almost}) and relaxations (\cite{ef2x,amanatidis2021maximum,aram22,berger2021almost,Caragiannis2023,caragiannis2019envy,chasmjahan23,mahara2021extension}) of $\efx$ have become an important line of research in discrete fair division.

Another relaxation of envy-freeness proposed in discrete fair division literature is that of {\em envy-freeness up to some item} ($\ef$1), introduced by \cite{budish2011combinatorial}. It requires that each agent prefers her own bundle to the bundle of any other agent after removing some item from the latter. $\ef$1 allocations always exist and can be computed efficiently (\cite{lipton2004approximately}). \emph{Epistemic} $\efx$ is another relaxation of $\efx$ that was recently introduced by \cite{Caragiannis2023}, where they showed its existence and polynomial-time tractability for additive valuations. A recent work \cite{akrami2025epistemic} then established the existence of epistemic $\efx$ allocations for monotone valuations.

Following the work of \cite{christodoulou2023fair}, recent works have started to focus on $\efx$ and $\ef1$ orientations and allocations on graph setting. \cite{landscape24} studies the mixed manna setting with both goods and chores and proves that determining the existence of $\efx$ orientations on simple graphs for agents with additive valuations is $\np$-complete and provides certain special cases like trees, stars, and paths where it is tractable. \cite{zeng2024structure} relates the existence of $\efx$ orientations and the chromatic number of the graph. Recently, \cite{deligkas2024ef1} showed that $\ef1$ orientations always exist for monotone valuations and can be computed in pseudo-polynomial time.

\emph{Proportionality} (\cite{dubins1961cut,steinhaus1948problem}) and {\em maximin fair share} (\cite{budish2011combinatorial}) are two other important fairness notions; we refer the readers to
an excellent recent survey by~\cite{survey2023} (and references within) for further discussion.

Research on fair division has also extended to the setting of \emph{indivisible chores}, where agents incur \emph{costs} rather than receive utilities. In this setting, substantial progress has been made toward understanding approximate $\efx$ fairness. The early work of \cite{ZW24} achieved an $O(n^2)$-approximation for additive cost functions, and subsequent improvements by \cite{GMQ25} tightened this to a constant $4$-$\efx$ guarantee for any number of agents. Parallel developments by \cite{CS24} and \cite{afshinmehr2024approximateefxexacttefx} established the existence of $2$-$\efx$ allocations for three agents, and \cite{garg2025existence2efxallocationschores} later generalized this result to arbitrary numbers of agents. Despite these advances, exact $\efx$ allocations remain known only in highly restricted settings. In fact, \cite{CS24} demonstrated that for monotone cost functions, some instances admit no $\efx$ allocation at all. Meanwhile, a long-standing open question on the compatibility of fairness and efficiency—whether an allocation can be both $\ef1$ and $\po$—was recently answered in the affirmative by \cite{mahara2025existencefairefficientallocation}.

\paragraph{Roadmap:} We begin by introducing notation and definitions in Section~\ref{sec:prelim}. In Section~\ref{sec:orientations}, we develop an exhaustive list of scenarios where $\efx$ orientations exist depending on two parameters related to multi-graphs for monotone valuations. In Section~\ref{sec:monotone}, we develop a pseudo-polynomial-time algorithm to compute $\efx$ allocations for bipartite multi-graphs with monotone valuations. Finally, in Section~\ref{sec:improve}, we extend our results and develop a pseudo-polynomial-time algorithm to compute $\efx$ allocations for multi-cycles with $\mms$-feasible valuations.
\section{Notation and Definitions}\label{sec:prelim}
For any positive integer $k$, we use $[k]$ to denote the set $\{1, 2, \ldots, k\}$. We consider a set $[m]$ of $m$ goods (items) that needs to be allocated among a set $[n] = \{1, 2, \ldots, n\}$ of $n$ agents in a fair manner. For ease of notation, we will use $g$ instead of $\{g\}$ for an item $g \in [m]$. 

\begin{definition}
    (Allocations). A partial allocation $X = (X_1, X_2, \ldots, X_n)$ is an ordered tuple of $n$ disjoint subsets of $[m]$, i.e., for every $i, j \in [n]$ we have  $X_i \cap X_j = \emptyset$. Here, $X_i$ denotes the bundle allocated to agent $i \in [n]$ in $X$. We say an allocation $X$ is complete if $\bigcup\limits_{i \in [n]} X_i = [m]$.
\end{definition}

\noindent
\textbf{Valuation Functions and Fairness Notions:} Each agent $i \in[n]$ specifies her preferences using a valuation function $v_i: 2^{[m]} \rightarrow R^{+}$, that assigns a non-negative value to every subset of items. We denote a fair division instance as $\mathcal{I} = \left<[n], [m], \{v_i\}_{i \in [n]}\right>$.

This work focuses on monotone, additive, cancelable, and $\mms$-feasible valuations that can be represented via \emph{multi-graphs}, defined below.

\begin{definition}
    (Monotone, Additive, Cancelable, and $\mms$-feasible Valuations). We say a valuation function $v: 2^{[m]} \rightarrow \mathbb{R}^+$ is monotone if for every $S \subseteq S' \subseteq [m]$, we have $v(S) \leq v(S')$. We say $v$ is \emph{additive} if for every subset $S \subseteq [m]$ of items, we have $v(S) = \sum\limits_{g \in S} v(g)$. 
    
    \noindent
    Next, we say $v$ is cancelable if for any $S, T \subseteq [m]$ and any $g \in [m] \setminus (S \cup T)$, we have
    $$v(S \cup \{g\}) > v(T \cup \{g\}) \implies v(S) > v(T).$$ That is, removing the same good from two bundles would not change the relative preference between the two. Note that the class of cancelable valuations is a strict superclass of additive valuations. 
    
    Finally, we say $v$ is $\mms$-feasible if for any $S \subseteq [m]$ and any partitions $A = (A_1, A_2)$ and $B = (B_1, B_2)$ of $S$, we have $\max(v(B_1), v(B_2)) \ge \min(v(A_1), v(A_2))$. 

    \noindent
    Note that the class of $\mms$-feasible valuations is a strict superclass of cancelable valuations.
\end{definition}

 Recently, \cite{christodoulou2023fair} proposed a valuation class that can be represented by graphs. Here, vertices correspond to agents and edges correspond to items that are valued positively only by the two agents at their endpoints. \citeauthor{christodoulou2023fair} studied simple graphs wherein there is at most one edge between every pair of adjacent vertices. In this work, we focus on a natural extension of these instances to \emph{multi-graphs} where we allow multiple edges between agents. 

\begin{definition}
    (Multi-graph Instances). A fair division instance $\mathcal{I}= \left<[n], [m], \{v_i\}_{i \in [n]}\right>$ on a multi-graph\footnote{A multi-graph can have multiple edges between two vertices.} is represented via a multi-graph $G = (V, E)$ where agents form the set $V = [n]$ of vertices and items form $m$ edges in $E$ with the following structure: for every agent $i \in [n]$ and every subset $S \subseteq [m]$, $v_i(S) = v_i(S \cap E(i))$, where $E(i)$ denotes the set of items incident to agent $i$, i.e., an item $g \in [m]$ has a positive marginal value for agent $i$ if and only if $g$ is incident to $i$. 

    % for every agent $i \in [n]$ and every item $g \in [m]$, $v_i(g) > 0$ if and only if $g$ is incident to $i$. 
\end{definition}

We denote the set of edges between agents $i$ and $j$ by $E(i, j)$. Note that $E(i, j) = E(j, i)$. Also, we use the words `agent' and `vertex' interchangeably, similarly for `item' and `edge'. 

\begin{definition} (Symmetric Instances)
    We say a multi-graph instance is symmetric if the valuations are additive and for any item $g \in E(i,j)$ with $i,j \in[n]$ is identically valued by both $i$ and $j$, i.e., $v_i(g)=v_j(g)$.
\end{definition}

Let us now define the concept of non-wasteful allocations known as \emph{orientations}.

\begin{definition}
    (Partial) Orientation). A partial orientation is a partial allocation where an item $g$ (if allocated) is given to an agent $i$ such that $g$ is incident to $i$ in the given multi-graph. This can be represented by directing the (allocated) edges in the graph towards the vertex receiving the edge.
\end{definition}

We use the popular notion of \emph{envy-freeness up to any good ($\efx$)} as a standard fairness notion in our work. Let us first define the concept of \emph{strong envy} and some useful constructs related to envy.

\begin{definition}
    (Envy and Strong Envy). Given an allocation $X = (X_1, X_2, \ldots, X_n)$, we say $i$ envies $j$ if $v_i(X_j) > v_i(X_i)$, and we say $i$ \emph{strongly envies} $j$ if there exists an item $g \in X_j$ such that  $v_i(X_i)<v_i(X_j \setminus g)$.
\end{definition}

\begin{definition}
    (Envy-Freeness Up to Any Good ($\efx$)). We say an allocation is $\efx$ if there is no strong envy between any pair of agents. Moreover, we say an allocation $X$ is $\alpha$-$\efx$ for an $\alpha \in (0, 1]$ if $v_i(X_i) \geq \alpha \cdot v_i(X_j \setminus g)$ for every $i,j \in [n]$ and $g \in X_j$.
\end{definition}

\begin{definition}
    (EFX-Feasibility). Given a partition $X = (X_1, X_2, \ldots, X_n)$ of items into $n$ bundles, we say bundle $X_k$ is EFX-feasible for agent $i$ if we have $v_i(X_k) \ge \max\limits_{j \in [n]} \max\limits_{g \in X_j} v_i(X_j \setminus g)$.
\end{definition}

We say a bundle containing one item as a ``singleton". Note that no agent strongly envies an agent owning a singleton. Also, since, in an orientation on a given multi-graph $G$, a vertex $i \in [n]$ receives edges that are incident to her, it implies that the set of $i$'s neighbors in $G$ are the only agents that can possibly strongly envy her. This leads to the following observation, which we will use frequently in future sections. 

\begin{observation} \label{obs:orientation_efx}
    A partial orientation is $\efx$ on a multi-graph if and only if no agent strongly envies her neighbor.
\end{observation}

Next, we demonstrate a property of $\efx$ orientations on multi-graphs via the following lemma.

\begin{restatable}{lemma}{lemenviedsingleton}\label{lem:envied-singleton}
    For a multi-graph instance, consider a partial $\efx$ orientation $X$ where a vertex $i$ is envied by one of her neighbors $j$. Then, we must have $X_i \subseteq E(i,j)$. In particular, any vertex is envied by at most one neighbor in any $\efx$ orientation.
\end{restatable}
\begin{proof}
    For contradiction, let us assume that there exists an agent $i$, who is envied by her neighbor agent $j$ such that there exists an $e \in X_i$, where $e \in E(i,k)$ for some $k \neq j$. Note that $v_j(X_j) < v_j(X_i) = v_j(X_i \setminus e)$, contradicting  the fact that $X$ is $\efx$.
\end{proof}

\iffalse
\begin{lemma} \label{lem:envied-singleton}
For a multi-graph instance, consider a partial $\efx$ orientation $X$ where a vertex $i$ is envied by one of her neighbors $j$. Then, we must have $X_i \subseteq E(i,j)$. In particular, any vertex is envied by at most one neighbor in any $\efx$ orientation.
\end{lemma}

\begin{proof}
    For contradiction, let us assume that there exists an agent $i$, who is envied by her neighbor agent $j$ such that there exists an $e \in X_i$, where $e \in E(i,k)$ for some $k \neq j$. Note that $v_j(X_j) < v_j(X_i) = v_j(X_i \setminus e)$, contradicting  the fact that $X$ is $\efx$.
\end{proof}

\fi

%Note that since we focus on multi-trees, the envy graph of any partial orientation does not contain envy cycles. Also, by Lemma~\ref{lem:envied-singleton}, every vertex is envied by at most one neighbor; thus, the envy path starting at an envied vertex is unique.

\subsection{Graph Theory Definitions}

Our work characterizes the existence of $\efx$ orientations based on the parameter $q$ (the maximum number of edges between any of agents in $G$) and the diameter of the multi-graph. For every multi-graph instance, We begin by defining some useful notions related to a multi-graph. 
\begin{definition}
    (Skeleton of a Multi-graph). For a multi-graph $G = (V, E)$, we define its skeleton as a graph $G' = (V, E')$  where $G'$ has the same set of vertices, and there is a single edge between two vertices if they share at least one edge in $G$, i.e., $i$ is connected to $j$ in $G'$ if $E(i, j) \neq \emptyset$ in $G$.
\end{definition}

\begin{definition}
    ($d(G)$ and $q$ of a Multi-graph $G$). We define $d(G)$ as the diameter of $G$, which is the length of the longest shortest path in the skeleton of $G$. And, we denote $q = \max_{i,j \in [n]} |E(i,j)|$ to be the maximum number of edges between any pair of agents in $G$.
\end{definition}

\begin{definition}
    (Center of a Multi-graph). For a multi-graph $G$, center $c \in V$ is a member of the set $\mathrm{argmin}_{x \in V} \max_{v \in V}d(x,v)$, where $d(x,v)$ is the distance between $x$ and $v$ in $G$. In this paper, we choose one arbitrarily if we have multiple centers.
\end{definition}

In this work, we focus on bipartite multi-graphs, that we define next.

\begin{definition}
    (Bipartite Multi-graph). A bipartite multi-graph $G = (V,E)$ has a skeleton that is a bipartite graph. We denote $V = S \sqcup T$ with two partitions $S$ and $T$ having no edge between them.
\end{definition}

\begin{definition}
    (Multi-star, Multi-$P_n$, Multi-cycle, and Multi-tree). A multi-H has a skeleton that is an $H$ graph, where $H$ can be a star, a path $P_n$ of length $n - 1$, a cycle, or a tree.
\end{definition}

Note that bipartite multi-graphs is an important class consisting of multi-trees, multi-cycles of even length, and planer multi-graphs with faces of even length, to name a few.

\section{EFX Orientations on Bipartite Multi-graphs} \label{sec:orientations}
 \citeauthor{christodoulou2023fair} shows that $\efx$ orientations may not always exist, even on simple graphs. Therefore, it is not surprising when we show the same on bipartite multi-graphs. In particular, we examine multi-cycles (a special kind of bipartite multi-graphs) and show that even for $q = 2$, $\efx$ orientations may not always exist for four agents; see Figure~\ref{fig:counter-example_c4}. Nonetheless, we identify the correct parameters to characterize the scenarios where $\efx$ orientations are guaranteed to exist for monotone valuations; see Table~\ref{tab:orientation}. 

 To do so, we proceed step by step, carefully considering all possible cases. Initially, we focus on bipartite multi-graphs with diameters, $d(G)$, of small numbers. For $d(G)=1$, bipartite multi-graphs become multi-$P_2$, for which an $\efx$ orientation can easily be achieved using the cut-and-choose protocol. As a warm-up, we show the existence of $\efx$ orientations for multi-trees with $d(G)=2$, i.e., for multi-stars. Here, we do so for $q = 2$, but the same approach can be generalized to any $q$, which we discuss in the next section (in Proposition~\ref{prop:multi-stars}).

In this section, we present counter-examples for various cases, via figures, where we have symmetric instances, and the number on each edge depicts the value of that edge for both endpoints.

\begin{restatable}{proposition}{multistarprop}\label{prop:2}
    $\efx$ orientations exist on multi-stars for $q = 2$.
\end{restatable}
\begin{proof}
        Let $r$ be the center of the given multi-star instance. Every leaf $i$ chooses the edge she prefers the most from the two edges in $E(i, r)$, and we allocate the remaining edges to agent $r$, hence creating an orientation. By Observation~\ref{obs:orientation_efx}, note that to prove that the allocation is $\efx$, we only need to check if the $\efx$ condition is satisfied between agent $r$ and her neighbors. Agent $r$ does not strongly envy any leaf because every leaf has a singleton edge in her bundle. Moreover, every leaf receives her preferred item from $E(i, r)$, so she will not envy agent $r$. This completes our proof.
\end{proof}

\iffalse
\begin{proposition} \label{prop:2}
    $\efx$ orientations exist on multi-stars for $q = 2$.
\end{proposition}

\begin{proof}
    Let $r$ be the center of the given multi-star instance. Every leaf $i$ chooses the edge she prefers the most from the two edges in $E(i, r)$, and we allocate the remaining edges to agent $r$, hence creating an orientation. By Observation~\ref{obs:orientation_efx}, note that to prove that it is $\efx$, we only need to check if the $\efx$ condition is satisfied between agent $r$ and her neighbors. Agent $r$ does not strongly envy any leaf because every leaf has a singleton edge in her bundle. And, every leaf receives her preferred item from $E(i, r)$, so it will not envy agent $r$. This completes our proof.
\end{proof}

\fi

However, an $\efx$ orientation might not always exist on bipartite multi-graphs with $d(G) \geq 2$. We prove it by providing an example in the following theorem; see Figure \ref{fig:counter-example_c4}.

\begin{restatable}{theorem}{multicyclethm}\label{thm:orientation-counter-example-cycle}
     $\efx$ orientations on cyclic bipartite multi-graphs with $d(G) \geq 2$ and any $q \geq 2$ may not exist (even on symmetric instances).
\end{restatable}
\begin{proof}
     We use Figure \ref{fig:counter-example_c4} to represent our counter-example (on multi-cycles with four agents) where the numbers (with $1 \gg \epsilon \gg \delta \approx 0$) on each edge show the value of that edge for both endpoints. We will try to construct an $\efx$ orientation and show that it is not possible.
    
    Without loss of generality, we assume the edge with value $10 + \frac{1}{2} \epsilon$ between agents $a$ and $b$ is allocated to agent $a$. First, let us suppose that the other edge in $E(a,b)$ is also allocated to $a$. For agent $b$ to not strongly envy $a$, the bundle $E(b,c)$ must be fully allocated to $b$. This leaves us no way to allocate items to $c$ so that she does not strongly envy $b$, reaching a contradiction. Therefore, the other edge in $E(a, b)$ must be allocated to $b$. 
    
    Now, consider the edge with value $10$ in $E(b,c)$. If we allocate it to $b$, then $c$ will strongly envy her. Therefore, this edge must be allocated to $c$. Observe that no other edge can be allocated to $c$; otherwise, $b$ will strongly envy $c$. The same argument goes for $a$, and hence, it can only have one item of value $10 + \frac{1}{2} \epsilon$. Therefore, $d$ must receive all her incident edges, making $a$ strongly envy $d$. Therefore, no $\efx$-orientation exists for this instance.

    It is easy to observe that the above counter-example can be extended by adding more edges with value $\delta$ between $c$ and $d$ or adding more vertices to the graph with edges of value $\delta' \approx 0$ to achieve a counter-example for bipartite multi-cyclic graphs with $q \geq 2$, and $d(G) \geq 2$.
\end{proof}

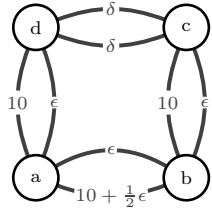
\begin{figure}
    \centering
    \begin{tikzpicture}
    \Vertex[label = a, x = 2, y = 0, color = white]{1}
    \Vertex[label = b, x = 4, y = 0, color = white]{2}
    \Vertex[label = c, x = 4, y = 2, color = white]{3}
    \Vertex[label = d, x = 2, y = 2, color = white]{4}
    
    \Edge[label = $\epsilon$, bend=30](1)(2)
    \Edge[label = $10 + \frac{1}{2} \epsilon$, bend=20](2)(1)
    \Edge[label = $10$, bend=20](2)(3)
    \Edge[label = $\epsilon$, bend=20](3)(2)
    \Edge[label = $10$, bend=20](1)(4)
    \Edge[label = $\epsilon$, bend=20](4)(1)
    \Edge[label = $\delta$, bend=20](3)(4)
    \Edge[label = $\delta$, bend=20](4)(3)
    
    \end{tikzpicture}
    
    \caption{Counter-example for Theorem~\ref{thm:orientation-counter-example-cycle}.}
    %Example for the non-existence of $\efx$ orientations for bipartite multi-graph (multi-cycle) with $n=4$ and $q = 2$.
    \label{fig:counter-example_c4}
\end{figure}

Since $\efx$ orientations do not exist even on bipartite multi-cyclic graphs with small diameters, we examine the existence of $\efx$ orientation on multi-trees as an important subset of bipartite multi-graphs. We show $\efx$ orientations exist on multi-trees with $d(G) \leq 4$ and $q = 2$.

\begin{restatable}{theorem}{orietationdiameterexistthm}\label{thm:orientation-P5}
        $\efx$ orientations exist for monotone valuations and can be computed in polynomial time on multi-trees with $d(G)\leq 4$ and $q = 2$.
\end{restatable}

\begin{proof}
    Let us denote the center of $G$ as $c$ and let it be rooted at $c$. We proceed via strong induction on the number of agents at depth $2$ in $G$.  We begin by making the induction statement stronger and prove that $\efx$ orientations exist on multi-trees with $d(G) \leq 4$ and $q=2$ while satisfying the following two properties:

    \begin{enumerate}
        \item For any envied agent $i$ at depth one, the set $E(c, i)$ is completely allocated to either $i$ or $c$. 
        \item For every envied agent $i$ at depth one, $i$ does not envy $c$.
    \end{enumerate}

    We consider multi-stars for the base case of the induction where there are zero agents at depth 2. Let $c$ pick her most valuable item. Then, we allocate all the remaining items incident to $c$ to their other endpoint. In this orientation, properties (1) and (2) are satisfied. Also, the orientation is $\efx$ because $c$ owns a single item, so no one strongly envies her, and $c$ does not strongly envy anyone because she receives her most preferred item. 
    
    Now, we move to the induction step. For every vertex $t$ at depth one in $G$, let $f_t$ be her most valued item in the set $\bigcup\limits_{j \text{ is a child of } t \ \text{in} \ G} E(t, j)$. For the induction step, consider an arbitrary vertex $i$ at depth one and remove all her children from $G$. According to the induction hypothesis, an $\efx$ orientation $X$ exists for the remaining multi-tree that satisfies properties (1) and (2). Now, we add the removed vertices back and show how we extend $X$ to an $\efx$ orientation for $G$ (by allocating the items in $\bigcup\limits_{j \text{ is a child of } i \ \text{in} \ G} E(i, j)$). We distinguish between three possible cases:

    \begin{itemize}
        \item \textbf{$i$ is non-envied in $X$:} Here, we let every child $j$ of $i$ choose her preferred item in $E(i, j)$ and allocate the other edge to agent $i$. Agent $i$ does not strongly envy her children because they own singletons. This gives an $\efx$ orientation for $G$. Note that agent $i$ remains non-envied, and using the induction hypothesis, we satisfy properties (1) and (2) here as well.
        
        \item \textbf{$i$ is envied in $X$ and $v_i(E(c, i)) \geq v_i(f_i)$:} Here, $i$ is envied by $c$ and by property (1), bundle $E(c, i)$ is allocated to $i$ in $X$. For every child $j$ of $i$, allocate the bundle $E(i,j)$ to agent $j$. Agent $i$ does not strongly envy her children because she receives $E(c, i)$, which values at least as much as $f_i$. This gives an $\efx$ orientation for $G$. Agent $i$ is envied, but we didn't reorient the edges in $X$; therefore, properties (1) and (2) are still satisfied using the induction hypothesis.

        \item \textbf{$i$ is envied in $X$ and $v_i(E(c, i)) < v_i(f_i)$:} Here, $i$ is envied by $c$ and by property (1), bundle $E(c, i)$ is allocated to $i$ in $X$. Assume $f_i \in E(i, j_0)$, where $j_0$ is a child of $i$ in $G$. Allocate $f_i$ to $i$ and give the other item in $E(i, j_0)$ to $j_0$. For every other child $j \neq j_0$ of $i$, allocate the bundle $E(i, j)$ to $j$. Reorient the edges in $E(c, i)$ and allocate all of them to $c$. 
        
        Now, if agent $c$ was non-envied in $X$, our resulting orientation will be $\efx$ too. Also, one can observe that both properties are satisfied as well. Thus, we only consider the case where $c$ was envied in $X$. By property (2), she can only be envied by a non-envied vertex at depth one. Let $h$ be the vertex that was envying $c$ in $X$. Reverse the orientation of edges between $c$ and $h$. Now, $i$ does not envy $c$. Also, $c$ does not envy $h$ because she received a more valuable bundle than her previous one. Therefore, $h$ is still non-envied. One can observe that the resulting orientation is again $\efx$, and both properties are satisfied.
    \end{itemize}

    Overall, we can obtain an $\efx$ orientation in $G$. Note that since our inductive step can be performed in polynomial time, we can compute an $\efx$ orientation for $G$ in polynomial time as well.
\end{proof}

Unfortunately, $\efx$ orientations may not exist on multi-trees with a greater diameter or higher $q$, as shown in the following theorems.

\begin{restatable}{theorem}{multipaththm}\label{thm:orientation-counter-example-bigq}
For multi-trees with $d(G) \geq 3$, $\efx$ orientations may not exist for $q \geq 3$ (even on symmetric instances).    
\end{restatable}
\begin{proof}
        We divide the counter-examples into two cases and show the non-existence of $\efx$ orientations even for four agents with different $q$ values (they are, in fact, multi-$P_4$ instances). We use Figure \ref{fig:counter-example-orientation-n} to represent our counter-examples. 
    
    \textbf{Case 1. $q = 3$:} Consider the instance in Figure \ref{fig:counter-example-orientation-n} (a) and assume, for contradiction, that it admits an $\efx$ orientation. Without loss of generality, let us suppose that agent $b$ receives the edge $(b, c)$. Agent $a$ cannot receive the set $E(a, b)$ completely; otherwise, $b$ will strongly envy $a$. Therefore, agent $b$ receives at least one edge other than $(b, c)$. Let us call this edge as $e$. The most valuable bundle that $c$ can receive such that $d$ does not strongly envy her is valued $2 + \epsilon$ to her. Therefore, we have,
    \begin{align*}
        v_c(X_b \backslash e) = v_c(E(b,c)) = 2 + \frac{3}{2} \epsilon > 2 + \epsilon \geq v_c(X_c),
    \end{align*}
    which is a contradiction.

    \textbf{Case 2. $q \geq 4$:} Consider the instance in Figure \ref{fig:counter-example-orientation-n} (b) where there are $q$-many edges with value $1$ for both endpoints between agents $a$ and $b$ and agents $c$ and $d$. We assume, for contradiction, that this instance admits an $\efx$ orientation. Again, without loss of generality, let agent $b$ receive the edge $(b,c)$. Agent $a$ cannot receive the set $E(a,b)$ completely; otherwise, $b$ will strongly envy $a$. Therefore, agent $b$ receives at least one edge other than $(b,c)$. Let us call this edge as $e$. The most valuable bundle that $c$ can receive such that $d$ does not strongly envy her is valued $\lceil \frac{q}{2} \rceil$ to her. Therefore, we have,
    \begin{align*}
        v_c(X_b \backslash e) = v_c(E(b,c)) = \lceil \frac{q}{2} \rceil + \epsilon > \lceil \frac{q}{2} \rceil \geq v_c(X_c),
    \end{align*}
    which is a contradiction.

    One can extend these counter-examples by adding nodes to the graph with edges with value $\delta \approx 0$ to achieve a counter-example for any $q \ge 3$ and $d(G) \geq 3$.
\end{proof}

\begin{figure}
    \centering
    \begin{subfigure}{0.51\textwidth}
        \centering
        \begin{tikzpicture}
            \Vertex[label = a, x = 2, y = 2, color = white]{1}
            \Vertex[label = b, x = 4, y = 0, color = white]{2}
            \Vertex[label = c, x = 6, y = 0, color = white]{3}
            \Vertex[label = d, x = 8, y = 2, color = white]{4}
            
            \Edge[label = $1$, bend=40](1)(2)
            \Edge[label = $1 + \epsilon$, bend=10](1)(2)
            \Edge[label = $1 + \epsilon$, bend=-20](1)(2)
    
            \Edge[label = $2 + \frac{3}{2} \epsilon$](2)(3)
            
            \Edge[label = $1$, bend=40](4)(3)
            \Edge[label = $1 + \epsilon$, bend=10](4)(3)
            \Edge[label = $1 + \epsilon$, bend=-20](4)(3)
        \end{tikzpicture}
    \caption{}
    \end{subfigure}
    \hspace{0pt} 
    \begin{subfigure}{0.51\textwidth}
    \centering
        \begin{tikzpicture}
            \Vertex[label = a, x = 2, y = 2, color = white]{1}
            \Vertex[label = b, x = 4, y = 0, color = white]{2}
            \Vertex[label = c, x = 6, y = 0, color = white]{3}
            \Vertex[label = d, x = 8, y = 2, color = white]{4}
    
            \path (1) -- (2) node [black, font=\normalsize, midway, sloped, rotate=90] {$\dots$};
            \path (3) -- (4) node [black, font=\normalsize, midway, sloped, rotate=270] {$\dots$};
            
            \Edge[label = $1$, bend=50](1)(2)
            \Edge[label = $1$, bend=20](1)(2)
            \Edge[label = $1$, bend=-30](1)(2)
            \Edge[label = $1$, bend=-60](1)(2)
    
            \Edge[label = $\lceil \frac{q}{2} \rceil + \epsilon $](2)(3)
            
            \Edge[label = $1$, bend=50](4)(3)
            \Edge[label = $1$, bend=20](4)(3)
            \Edge[label = $1$, bend=-30](4)(3)
            \Edge[label = $1$, bend=-60](4)(3)
        \end{tikzpicture}
    \caption{}
    \end{subfigure}

    \caption{Counter-example for Theorem~\ref{thm:orientation-counter-example-bigq}.} 
    %Multi-$P_4$ instances with $q=3$ and with $q \geq 4$ that does not admit any $\efx$ orientation. There are $q$ edges with value $1$ for both endpoints between agents $a$ and $b$ and agents $c$ and $d$.}
    \label{fig:counter-example-orientation-n}
\end{figure}
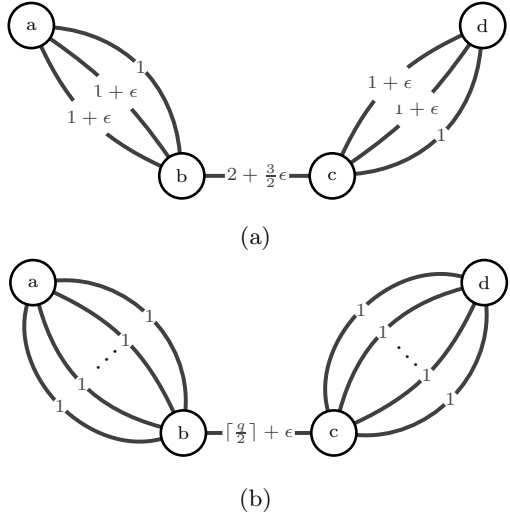

\begin{theorem} \label{thm:orientation-counter-example-q2}
For multi-trees with $d(G) \geq 5$, $\efx$ orientations may not exist even for $q = 2$ (even on symmetric instances).
\end{theorem}

\begin{proof}
    We present a multi-$P_6$ instance with $q = 2$ that does not admit any $\efx$ orientation. We use Figures~\ref{fig:counter-example_p3_aamas} and \ref{fig:counter-example_p6} to represent our counter-example.
    
    We first construct a special multi-$P_3$ instance where a specific node in any of its $\efx$ orientations is envied. Consider the multi-$P_3$ instance in Figure~\ref{fig:counter-example_p3_aamas}(a). One can easily show that it admits exactly two $\efx$ orientations, and agent $3$ is envied in both; see Figure~\ref{fig:counter-example_p3_aamas}(b) and Figure~\ref{fig:counter-example_p3_aamas}(c), 

    Now we build the multi-$P_6$ instance in Figure~\ref{fig:counter-example_p6} that is made by two copies of the above-mentioned special multi-$P_3$ instance. We connect them with an edge of value $\delta \ll \epsilon$ using the two vertices (vertex $3$ and vertex $4$) that are always envied in $\efx$ orientations of the multi-$P_3$ parts. Without loss of generality, let agent $3$ receive the edge $(3,4)$ with value $\delta$. Now, since in any $\efx$ orientation of the special multi-$P_3$ instance, agent $3$ was always envied, the addition of edge with value $\delta$ to agent $3$'s bundle will therefore create strong envy against her.

    One can extend the above counter-example by adding nodes to the graph with edges having value $\delta' \approx 0$ to achieve a counter-example for any $d(G) \geq 5$.
\end{proof}

\begin{figure}
    \centering
    \begin{subfigure}[t]{0.25\textwidth}
        \centering
        \begin{tikzpicture}[baseline=(current bounding box.north)]
        \Vertex[label = 1, x = 2, y = 0, color = white]{1}
        \Vertex[label = 2, x = 4, y = 0, color = white]{2}
        \Vertex[label = 3, x = 6, y = 0, color = white]{3}
        
        \Edge[label = $\epsilon$, bend=40](1)(2)
        \Edge[label = $10 + \frac{1}{2} \epsilon$, bend=40](2)(1)
        \Edge[label = $10$, bend=40](2)(3)
        \Edge[label = $\epsilon$, bend=40](3)(2)
        \end{tikzpicture}
        \caption{}
    \end{subfigure}
    %\hspace{15pt} 
    
    \begin{subfigure}[t]{0.25\textwidth}
        \centering
        \begin{tikzpicture}[baseline=(current bounding box.north)]
        \Vertex[label = 1, x = 2, y = 0, color = white]{1}
        \Vertex[label = 2, x = 4, y = 0, color = white]{2}
        \Vertex[label = 3, x = 6, y = 0, color = white]{3}
        
        \Edge[label = $\epsilon$, bend=40, Direct](1)(2)
        \Edge[label = $10 + \frac{1}{2} \epsilon$, bend=40, Direct](2)(1)
        \Edge[label = $10$, bend=40, Direct](2)(3)
        \Edge[label = $\epsilon$, bend=40, Direct](3)(2)
        \end{tikzpicture}
        \caption{}
    \end{subfigure}
    \hspace{40pt} 
    \begin{subfigure}[t]{0.25\textwidth}
        \centering
        \begin{tikzpicture}[baseline=(current bounding box.north)]
        \Vertex[label = 1, x = 2, y = 0, color = white]{1}
        \Vertex[label = 2, x = 4, y = 0, color = white]{2}
        \Vertex[label = 3, x = 6, y = 0, color = white]{3}
        
        \Edge[label = $10 + \frac{1}{2} \epsilon$, bend=-40, Direct](1)(2)
        \Edge[label = $\epsilon$, bend=-40, Direct](2)(1)
    
        \Edge[label = $\epsilon$, bend=-40, Direct](2)(3)
        \Edge[label = $10$, bend=40, Direct](2)(3)
        \end{tikzpicture}
        \caption{}
    \end{subfigure}
    \caption{We depict a multi-$P_3$ instance whose every $\efx$ orientation leaves agent $3$ envied. We use this instance as a building block to give a multi-$P_6$ instance that does not admit any $\efx$ orientation. Orientations (b) and (c) are the only $\efx$ orientations in this instance.}
    \label{fig:counter-example_p3_aamas}
\end{figure}
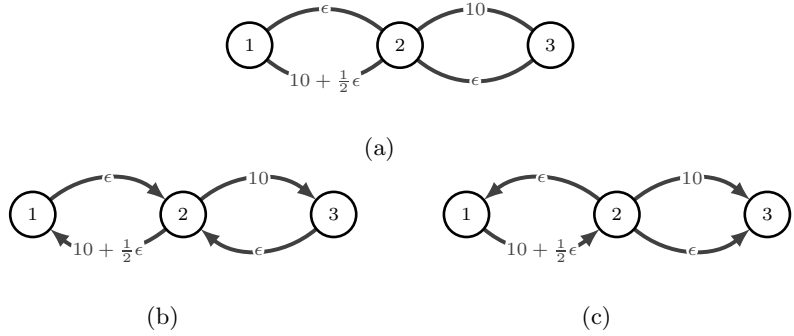

\begin{figure}
    \centering
    \begin{tikzpicture}
    \Vertex[label = 1, x = 2, y = 0, color = white]{1}
    \Vertex[label = 2, x = 4, y = 0, color = white]{2}
    \Vertex[label = 3, x = 6, y = 0, color = white]{3}
    \Vertex[label = 6, x = 12, y = 0, color = white]{6}
    \Vertex[label = 5, x = 10, y = 0, color = white]{5}
    \Vertex[label = 4, x = 8, y = 0, color = white]{4}
    
    \Edge[label = $\epsilon$, bend=40](1)(2)
    \Edge[label = $10 + \frac{1}{2} \epsilon$, bend=40](2)(1)
    \Edge[label = $10$, bend=40](2)(3)
    \Edge[label = $\epsilon$, bend=40](3)(2)
    \Edge[label = $\delta$](3)(4)
    \Edge[label = $\epsilon$, bend=40](6)(5)
    \Edge[label = $10 + \frac{1}{2} \epsilon$, bend=40](5)(6)
    \Edge[label = $10$, bend=40](5)(4)
    \Edge[label = $\epsilon$, bend=40](4)(5)
    
    \end{tikzpicture}
    
    \caption{A multi-$P_6$ instance that does not admit any $\efx$ orientation.}
    \label{fig:counter-example_p6}
\end{figure}
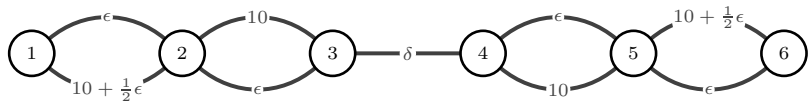

\subsection{Hardness of Deciding the Existence of EFX Orientations}

In this section, we consider the computational problem of deciding whether a given instance on a bipartite multi-graph (even a multi-tree with a constant number of agents) admits an $\efx$ orientation. We reduce the $\np$-complete \emph{Partition Problem} to our problem.

\paragraph{Partition Problem.} Consider a multi-set\footnote{A multi-set allows multiple instances for each of its elements.} $P= \{p_1, p_2, \ldots, p_k\}$ of $k$ non-negative integers. The problem is to decide whether $P$ can be partitioned into two multi-sets $P_1$ and $P_2$ such that $\sum_{p \in P_1}p = \sum_{p \in P_2}p$.

We prove a stronger claim and prove hardness for multi-tree instances. 

\begin{theorem}\label{thm:hardness-orientation-multi-tree}
    The problem of deciding whether a fair division instance on a multi-tree (with additive valuations) admits an $\efx$ orientation is $\np$-complete. It holds true even for symmetric instances with a constant number of agents.

\end{theorem}

\begin{proof}
    Given an orientation, it is easy to verify if it is $\efx$. Hence, the problem belongs in $\np$. 
    
    Let us now consider an instance $P = \{p_1, p_2, \ldots, p_k\}$ of the Partition problem. We will construct a fair division instance on a multi-tree with eight vertices, as depicted in Figure \ref{fig:NP-hardness}. Note that we have used the multi-$P_3$ instance used in the proof of Theorem~\ref{thm:orientation-counter-example-q2} here as well. Following the similar lines, we can argue that in any $\efx$ orientation in this instance, agent 2 envies agent 3, and agent 7 envies agent 6. Thus, the two edges $(3, 4)$ and $(5, 6)$ are allocated to agents 4 and 5 respectively. Now, one can observe that this instance admits an $\efx$ orientation if and only if the set $P$ can be partitioned into two sets of equal sum. This completes our proof.
\end{proof}

\begin{figure}
    \centering
    \begin{tikzpicture}

        \Vertex[label = 1, x = 0, y = 0, color = white]{1}
        \Vertex[label = 2, x = 2, y = 0, color = white]{2}
        \Vertex[label = 3, x = 4, y = 0, color = white]{3}

        \Vertex[label = 4, x = 5, y = 1, color = white]{4}
        \Vertex[label = 5, x = 7, y = 1, color = white]{5}
        
        \Edge[label = $\epsilon$, bend=40](1)(2)
        \Edge[label = $10 + \frac{1}{2} \epsilon$, bend=40](2)(1)
        \Edge[label = $10$, bend=40](2)(3)
        \Edge[label = $\epsilon$, bend=40](3)(2)

        \Vertex[label = 6, x = 8, y = 0, color = white]{6}
        \Vertex[label = 7, x = 10, y = 0, color = white]{7}
        \Vertex[label = 8, x = 12, y = 0, color = white]{8}
        
        \Edge[label = $\epsilon$, bend=40](7)(8)
        \Edge[label = $10 + \frac{1}{2} \epsilon$, bend=40](8)(7)
        \Edge[label = $10$, bend=40](6)(7)
        \Edge[label = $\epsilon$, bend=40](7)(6)

        \path (4) -- (5) node [black, font=\normalsize, midway, sloped, rotate=270] {$\dots$};

        \Edge[label = $p_1$, bend = 60](4)(5)
        \Edge[label = $p_2$, bend = 30](4)(5)
        \Edge[label = $p_{k-1}$, bend = -30](4)(5)
        \Edge[label = $p_k$, bend = -60](4)(5)

        \Edge[label = $\delta$](3)(4)
        \Edge[label = $\delta$](5)(6)

    \end{tikzpicture}
    \caption{The construction used in proof of Theorem~\ref{thm:hardness-orientation-multi-tree}. Here, $\delta \ll \epsilon \ll 1$.}
    \label{fig:NP-hardness}
\end{figure}
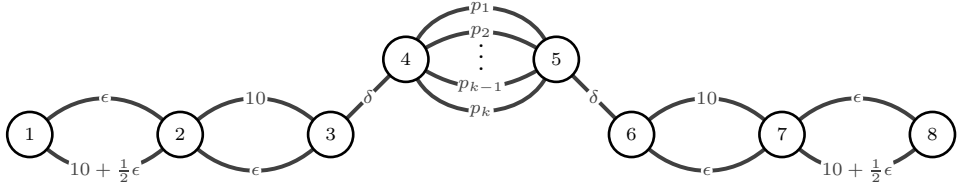

\iffalse
Now, as a corollary, we have the following theorem:

\begin{corollary}
    Deciding whether an EFX orientation exists on a multi-bipartite graph is NP-complete, even when the graph has a constant number of vertices and the valuations are additive.
\end{corollary}
\fi

Theorem~\ref{thm:hardness-orientation-multi-tree} immediately implies the following, which is independently proved by \cite{deligkas2024ef1}:

\begin{corollary} \label{cor:orientation}
    Deciding whether an EFX orientation exists on a multi-graph with additive valuations is $\np$-complete, even for symmetric instances with a constant number of agents.
\end{corollary}

\section{Existence of EFX on Bipartite Multi-Graphs for Monotone Valuations} \label{sec:monotone}

In this section, we prove our main result showing that any fair division instance on a bipartite multi-graph with \emph{monotone} valuations always admits an  $\efx$ allocation. Furthermore, such allocations can be computed in pseudo-polynomial time for \emph{monotone} valuations and in polynomial time for \emph{cancelable} valuations, a strict superclass of additive valuations (Theorem~\ref{thm:multi-bipartite-main}). Recall that (as discussed in Section~\ref{sec:orientations}), $\efx$ orientations may not exist for these instances, and hence, we know that an $\efx$ allocation would allocate some items to an agent who may not value them. 

We will denote a fair division instance on a bipartite multi-graph by $G = (S \sqcup T, E)$, where $S$ and $T$ represent the two bi-partition parts of its skeleton. We begin by discussing the main idea of our technique (in Section~\ref{sec:allocations-idea}) and then define some concepts and properties (in Section~\ref{sec:definition-4}) useful for our proof (in Sections~\ref{sec:allocations-proof1}, \ref{sec:allocations-proof2}, \ref{sec:allocations-proof3}, \ref{sec:safe-set}, and \ref{sec:complete-allocation}). For the ease of understanding, we have further presented an example of executing our algorithm step by step on a bipartite multi-graph with additive valuations (see Appendix \ref{appendixA}).

\subsection{Main Idea} \label{sec:allocations-idea}

We introduce the concept of \emph{configuration} to decide how to allocate the edges between any two adjacent vertices (in $G$) to their endpoints. It has a flavor that is similar to the cut-and-choose protocol (used for finding $\efx$ allocations between two agents). In our proof, we will use this configuration to partially orient the edges between two agents. We define it as follows:

\begin{definition}\label{def:configuration} 
\textbf{\emph{$T$-cut} Configurations:} For any pair of agents $i\in S$ and $j \in T$, we let agent $j$ to partition the set $E(i, j)$ into two bundles $C_1$ and $C_2$ such that both are $\efx$-feasible for $j$ (for the items $E(i, j)$). We call the partition $(C_1, C_2)$ as the \emph{$j$-cut} configuration between agents $i$ and $j$. 
\end{definition}

We now briefly discuss the running time of computing such partitions in the following remark, which will later be used to derive the overall running time of our algorithm.

\begin{remark}\label{rem:conf-runtime}
    It is well known that the algorithm referred to as the PR algorithm (using the terminology as \cite{akrami2025efx}), first introduced by \cite{plaut2020almost}, can compute our cut configurations in pseudo-polynomial time for monotone valuations (see the proof of Theorem 3.49 in \cite{AGS25}). Moreover, \cite{AGS25} proved that with minor modifications, this algorithm can compute our desired cut configurations in polynomial time for cancelable valuations (see Lemma 4.7 in \cite{AGS25}).
\end{remark}

As a warm-up, we will use \emph{$T$-cut} configurations to prove the existence of $\efx$ orientations for multi-stars with any $q$. Previously, (in Proposition~\ref{prop:2}), we proved the existence of $\efx$ orientations for multi-stars with $q=2$.

\begin{restatable}{proposition}{multistarpropgeneral}\label{prop:multi-stars}
    $\efx$ orientations exist and can be computed in pseudo-polynomial time for multi-stars with any $q$. Moreover, we can compute such allocations in polynomial time for cancelable valuations.  
\end{restatable}
\begin{proof}
    Assume $c$ is the center of the given multi-star, and assume that $T = \{c\}$ in $G$. For every vertex $i$ incident to $c$, consider the \emph{$c$-cut} configuration $(C_1,C_2)$ between them. Then, we allocate the bundle which $i$ prefers between $C_1$ and $C_2$ to her and allocate the other bundle to $c$. Therefore, for every such vertex $i$, she will not envy agent $c$. Moreover, since we are using the \emph{$c$-cut} configurations, agent $c$ will not strongly envy any of her neighbors. Using Remark \ref{rem:conf-runtime}, it is clear that the final orientation can be computed in pseudo-polynomial time for monotone valuations and in polynomial time for cancelable valuations.
\end{proof}

\subsection{Some Useful Notions and Properties}\label{sec:definition-4}

For a partial allocation $X = (X_1, X_2, \ldots, X_n)$, we define the following new notions:

\begin{itemize}
    \item For any two adjacent agents $i \in S$ and $j \in T$, we define $A_{i,j}(X)$ as the set of \emph{available} edges in $E(i, j)$ for $i$ and $A_{j,i}(X)$ as the set of \emph{available} edges in $E(i, j)$ for $j$. Formally, we define these sets as specified below in Table \ref{tab:Ai,j defjintion}. Let us assume the $j$-cut configuration of $E(i,j)$ is $(C_1,C_2)$.

    \begin{table}[h!]
    \centering
    \begin{tabular}{|c|c|c|}
    \hline
    & $A_{i,j}(X)$ & $A_{j,i}(X)$\\
    \hline
    $E(i,j) \cap X_k = \emptyset \ \text{for all} \ k \in [n]$ & $\argmax\limits_{C_1, C_2} \{v_i(C_1), v_i(C_2)\}$ & $\argmax\limits_{C_1, C_2} \{v_j(C_1), v_j(C_2)\}$ \\
    \hline
    $X_i \cap E(i,j) = \emptyset, X_j \cap E(i,j) \neq \emptyset$  & $E(i, j) \setminus X_j$ & $\emptyset$ \\
    \hline
    $ X_i \cap E(i,j) \neq \emptyset, X_j \cap E(i,j) = \emptyset$ & $\emptyset$ & $E(i, j) \setminus X_i$\\
    \hline
    Any other case & Not defined & Not defined \\
    \hline
    \end{tabular}
    \caption{Definition of $A_{i,j}(X)$ and $A_{j,i}(X)$}\label{tab:Ai,j defjintion}
    \end{table}
    
    \item For $i \in [n]$, we define $A_i(X)$ to be her \emph{available set} of edges, i.e., $A_i(X) = \bigcup\limits_{j \neq i, j \in [n]} A_{i,j}(X)$.

    \item For $i \in [n]$, $U_i(X)$ is the set of all unallocated edges incident to $i$. Note that $A_i(X) \subseteq U_i(X)$.
    
    \item For $i \in [n]$, $B_i(X)$ is the set of all \emph{available bundles} for $i$,
    i.e. $B_i(X) = \{A_{i,j}(X) : j \neq i, j \in [n]\}$.

    \item For any envied agent $i \in [n]$, we define $S_i(X) \subseteq [n]$ to be her \emph{safe} set, as follows, $$S_i(X) = \{k \in [n]: k \ \text{is non-envied in} \ X \ \text{and} \ v_i(X_i) \ge v_i(X_k \cup A_i(X))\}.$$ That is, $i$ will not envy $k$ even if we allocate her whole available set $A_i(X)$ to $k$.
\end{itemize}

To achieve a complete $\efx$ allocation, we will use a approach that is structurally similar to  \cite{christodoulou2023fair}. We will first find a partial orientation with some nice properties and then allocate the remaining edges to some agent who is not incident to them. Identifying these key nice properties is a non-trivial challenge that we address next.

\paragraph{Key Properties:} We search for a partial allocation $X = (X_1, X_2, \ldots, X_n)$ with the following properties:

\begin{enumerate}
    \item $X$ is an $\efx$ orientation.

    \item For any two adjacent agents $i \in S$ and $j \in T$, items in $E(i, j)$ must be allocated based on the \emph{$j$-cut} configuration $(C_1,C_2)$ to either one of their endpoints. By this property, we mean that one of the following cases must happen for $X$ (following the Definition~\ref{def:configuration}):

    \begin{itemize}
        \item  Either $C_1 \subseteq X_i, C_2 \subseteq X_j$ or $C_2 \subseteq X_i, C_1 \subseteq X_j$.
        \item One of the bundles\footnote{$X_i$ can also have other items.} $C_1$ or $C_2$ is allocated to agent $i$, and the other bundle is unallocated in $X$.
        \item One of the bundles $C_1$ or $C_2$ is allocated to agent $j$, and the other bundle is unallocated in $X$.
        \item No item in the set $E(i, j)$ is allocated in $X$.
    \end{itemize}

    \item For any agent $i \in [n]$ and a set $B \in B_i(X)$, we have $v_i(X_i) \ge v_i(B)$.

    \item For any non-envied agent $i \in [n]$, we have $A_i(X) = \emptyset$.

    \item For any non-envied agent $i \in [n]$, we have $v_i(X_i) \ge v_i(U_i(X))$.

    \item For any envied agent $i \in [n]$, let $j$ envies $i$. Then, we have $j \in S_i(X)$.
        
\end{enumerate}

We are now finally equipped to present our algorithm. We will give a step-by-step procedure to satisfy each key property in the above order. These key properties ensure that there is an easy way to then convert a partial $\efx$ orientation to a complete $\efx$ allocation (see Section~\ref{sec:complete-allocation}).

As mentioned earlier, for the ease of understanding, we present an example of executing our algorithm step by step on a bipartite multi-graph with additive valuations (see Appendix \ref{appendixA}).

\subsection{Satisfying Properties (1)-(3)}
\label{sec:allocations-proof1}

We present a greedy algorithm that assigns a set of items to each agent and satisfies the first three properties. It works in the following manner.

Let $S = \{i_1, i_2, \ldots, i_{|S|}\}$ and $T = \{j_1, j_2, \ldots, j_{|T|}\}$ be the two bi-partitions in the given multi-graph $G$. We fix a picking sequence $\sigma = [i_1, \ldots, i_{|S|}, j_1, \ldots, j_{|T|}]$ that decides the order in which an agent comes and selects her most valuable available bundle. Since the definition of $A_{i, j}(X)$ is dynamic, the set of available bundles for some agents may change after another agent picks her favorite bundle in the picking sequence. Algorithm \ref{alg:greedy} illustrates this procedure. The properties of this algorithm are further formalized in Lemmas \ref{lem:greedy} and \ref{lem:envy-S}.

\begin{algorithm}[tb]
\caption{Greedy Orientation: Properties (1)-(3)}\label{alg:greedy}
\KwIn{A fair division instance on a bipartite multi-graph $G = (S \sqcup T, E)$}
\KwOut{A partial $\efx$ orientation satisfying properties $(1)$-$(3)$}
\BlankLine
\For{$(l \gets 1$;  $l \le |S|$ ;  $l  \pluseq  1)$}{

    $k \gets \argmax\limits_{k \in [n] \setminus \{i_l\}} v_{i_l}(A_{i_l, k}(X))$
    
    $X_{i_l} \gets A_{i_l, k}(X)$
  
}

\For{$(l \gets 1$;  $l \le |T|$ ;  $l  \pluseq  1)$}{

    $k \gets \argmax\limits_{k \in [n] \setminus \{j_l\}} v_{j_l}(A_{j_l, k}(X))$
    
    $X_{j_l} \gets A_{j_l, k}(X)$
  
}

\Return{$X = (X_1, \cdots, X_n)$}

\end{algorithm}

\begin{restatable}{lemma}{lemmagreedy}\label{lem:greedy}
        For a fair division instance on a bipartite multi-graph, the output allocation of Algorithm~\ref{alg:greedy} satisfies properties (1)-(3). Moreover, the algorithm runs in polynomial time if the cut configurations are given as input.
\end{restatable}
\begin{proof}
    Let us denote $X$ as the output allocation of Algorithm~\ref{alg:greedy}. First, note that $X$ is an orientation since every agent picks items adjacent to them. Next, since every agent selects a bundle greedily, properties (2) and (3) are trivially satisfied. Therefore, what is remaining is to argue that $X$ is $\efx$. By Observation~\ref{obs:orientation_efx}, we only need to check the strong envy between adjacent nodes. 
    
    Consider an arbitrary agent $i \in S$ and some agent $j \in T$ adjacent to $i$. Note that since $i$ precedes $j$ in the picking sequence,  $i$ does not envy $j$. If agent $i$ does not pick edges from the set $E(i, j)$, then $j$ clearly does not envy $i$. Now, consider the case where agent $i$ picks the bundle she prefers from the \emph{$j$-cut} configuration $(C_1,C_2)$ of $E(i, j)$. Let us assume  $v_i(C_1) \ge v_i(C_2)$ and hence $i$ picks $C_1$ (the other case is symmetric). Therefore, $C_2$ is available for $j$ on her turn to pick a bundle in Algorithm~\ref{alg:greedy}. Hence, $j$ receives a bundle such that $v_j(X_j) \geq v_j(C_2)$. By construction of $C_1$ and $C_2$, we know that for every item $g \in X_i=C_1$, we have $v_j(C_2) \ge v_j(C_1 \setminus g)$. Since, $v_j(X_j) \geq v_j(C_2)$, $j$ does not strongly envy $i$. Overall, we have shown that $X$ is a partial $\efx$ orientation, which completes our proof. It is also clear that the algorithm terminates in polynomial time if the cut configurations are given as input.
\end{proof}

\begin{restatable}{lemma}{lemenvyS}\label{lem:envy-S}
     In the output allocation $X$ of Algorithm~\ref{alg:greedy}, every envied vertex belongs to the set $S$.   
\end{restatable}
\begin{proof}
    We show that agents in $T$ are non-envied in $X$. Note that any agent in $S$ appears before any agent in $T$ in the picking sequence; therefore, they do not envy agents in $T$. Furthermore, agents in $T$ do not envy each other and value each other's bundle at zero since they are not connected in the graph skeleton, and our allocation is an orientation. Therefore, any envied vertex (if any) in $X$ must belong to the set $S$.
\end{proof}

As we proceed, Lemma \ref{lem:envy-S} will continue to hold while we obtain our desired orientation. As demonstrated in the following sections, we will not produce new envied vertices when we modify our allocation $X$ to satisfy properties (1)-(6).

\subsection{Satisfying Property (4)}\label{sec:allocations-proof2}
Let us now focus on satisfying property (4) that requires $A_i(X) = \emptyset$ for any non-envied agent $i \in [n]$. Let us assume that the output allocation $X$ of Algorithm~\ref{alg:greedy} violates property (4). Consider a non-envied agent $i \in [n]$ with $A_i(X) \neq \emptyset$. Therefore, an agent $j \in [n]$ exists such that $A_{i, j}(X) \neq \emptyset$. We will now allocate all of $A_{i,j}(X)$ either to $i$ or $j$, depending on the following three possible cases.\footnote{$A_{i, j}(X)$ will be allocated either to $i$ or $j$.}

\begin{itemize}
    \item \textbf{Case 1. A set of items in $E(i,j)$ is allocated to $j$:} Since $A_{i,j} \neq \emptyset$, by its definition and property (2), no edge in $E(i,j)$ is allocated to $i$. In this case, we can allocate $A_{i,j}(X)$, which is exactly the set $E(i, j) \setminus X_j$ to $i$. Since $X_j \cap E(i,j) \neq \emptyset$, $j$ chose the better bundle from the configuration of $E(i,j)$ during Algorithm~\ref{alg:greedy}. Also, the set $A_{i,j}(X)$ has value only to agents $i$ and $j$; therefore, since agent $i$ was non-envied before, the modified allocation remains $\efx$. One can observe that properties (1)-(3) remain satisfied. Observe that, in this case, $E(i, j)$ will be fully allocated.

    \item \textbf{Case 2. No item in $E(i,j)$ is allocated, and $j$ is non-envied:}  Without loss of generality, we can assume $i \in S$ and $j \in T$. Let the partition $(C_1, C_2)$ be the $j$-\emph{cut} configuration of the set $E(i, j)$. Let us assume $v_i(C_1) \ge v_i(C_2)$ (the other case is symmetric). Observe that by property (3) $v_i(X_i) \ge \max\{v_i(C_1), v_i (C_2)\}$ and $v_j(X_j) \ge \max\{v_j(C_1), v_j(C_2)\}$. Since the $X$ is an orientation, we have that $v_i(X_j) = v_j(X_i) = 0$. We now allocate $C_1$ to agent $i$ and $C_2$ to agent $j$ to obtain,
    \begin{align*}
        v_i(X_j \cup C_2) = v_i(C_2) \le v_i(X_i), \ \text{and} \  v_j(X_i \cup C_1) = v_j(C_1) \le v_j(X_j)
    \end{align*}

    Thus, the allocation remains $\efx$, and all the first three properties are still satisfied. Notice that $E(i, j)$ will be fully allocated in this case as well.

    \item \textbf{Case 3. No item in $E(i,j)$ is allocated and $j$ is envied:} 
    In this case, Lemma \ref{lem:envy-S} entails that $j \in S$ and $i \in T$. Let the partition $(C_1, C_2)$ be the \emph{$i$-cut} configuration of items $E(i,j)$. By property (3) of $X$, we have $v_j(X_j) \geq \max\{v_j(C_1), v_j(C_2)\}$. Assuming $v_i(C_1) \ge v_i(C_2)$ (the other case is symmetric), we allocate $C_1$ to agent $i$. Note that agent $j$ will not envy $i$ and agent $i$ remains non-envied in the modified allocation. Hence, the allocation remains $\efx$, and the properties (1)-(3) are still satisfied.
\end{itemize}

\paragraph{Formalized protocol (Algorithm \ref{alg:prop-4}) to satisfy property (4) along with properties (1)-(3):}  We repeat the following process as long as there is a non-envied agent $i \in [n]$ who violates property (4). We pick such a violator agent $i$. Then, for every agent $j \neq i$ such that $A_{i,j}(X) \neq \emptyset$, we allocate $A_{i,j}(X)$ according to the cases above. Note that we allocate at least one edge incident to $i$ at each step. Therefore, for each agent $i$, this step takes at most $O(m)$ iterations. Then, we repeat. At the end, for any non-envied agent $i \in [n]$, we ensure that $A_i(X) = \emptyset$, thereby satisfying property (4). Moreover, as discussed above, properties (1)-(3) remain satisfied as well. We abuse the notation and call the partial orientation we have built so far (that satisfies properties (1)-(4)) by $X$.

\begin{algorithm}[tb]
\caption{Allocating to Non-Envied Vertices: Properties (1)-(3) + Property (4)}\label{alg:prop-4}
\KwIn{A partial orientation $X$ that is an output of Algorithm~\ref{alg:greedy} on a bipartite multi-graph $G$, satisfying properties $(1)$-$(3)$}
\KwOut{A partial orientation $X$ satisfying properties $(1)$-$(4)$}

\BlankLine

\While{there exists a non-envied agent $i \in [n]$ such that $A_i(X) \neq \emptyset$}{

    \While{there exists an agent $j \in [n]$ such that $A_{i, j}(X) \neq \emptyset$}{

        \If{$X_j \cap E(i, j) \neq \emptyset$}{
        
            $X_i \gets X_i \cup A_{i, j}(X)$
            
        }
        
        \ElseIf{$j$ is non-envied}{

            Without loss of generality, let $i \in S$ and $j \in T$. 
            
            $(C_1, C_2) \leftarrow$ the \emph{$j$-cut} configuration of $E(i, j)$. 
            
            $C_{\ell} \leftarrow \argmax\limits_{C_1, C_2}\{v_i(C_1), v_i(C_2)\}$.

            $X_i \gets X_i \cup C_{\ell}$

            $X_j \gets X_j \cup C_{3-\ell}$
            
        }

        \Else{

            $(C_1, C_2) \leftarrow$ the \emph{$i$-cut} configuration of $E(i, j)$. 

            $C_{\ell} \leftarrow \argmax\limits_{C_1, C_2}\{v_i(C_1), v_i(C_2)\}$.

            $X_i \gets X_i \cup C_{\ell}$
            
        }
        
    }

}
 \Return{$X = (X_1, \dots, X_n)$}

\end{algorithm}

\begin{restatable}{claim}{claimunallocedges}\label{claim:unallocated_edges}
   After satisfying properties (1)-(4), if there exists a pair of agents $k, i \in [n]$ such that $A_{k,i}(X) \neq \emptyset$, then $k$ is an envied vertex, but $i$ is non-envied. Furthermore, $E(k,i) \setminus A_{k,i}(X)$ is allocated to $i$.
\end{restatable}
\begin{proof}
    Note that both $k$ and $i$ both cannot be envied by Lemma \ref{lem:envy-S}. Also, since $A_{k,i}(X) \neq \emptyset$, property (4) implies that agent $k$ cannot be non-envied. Thus, agent $k$ is envied, while agent $i$ is non-envied. By property (4), the set $E(k,i) \setminus A_{k,i}(X)$ must be allocated to agent $i$, otherwise $A_{i,k}(X) \neq \emptyset$, which is a contradiction.
\end{proof}

\begin{lemma}\label{lem:runtime-prop(4)}
    Algorithm \ref{alg:prop-4} terminates in polynomial time if the cut configurations are given as input. 
\end{lemma}
\begin{proof}
    For every non-envied agent $i \in [n]$ violating property (4), the outer while loop of Algorithm \ref{alg:prop-4} is executed, and it is clear that if the cut configurations are given as input, then $A_{i, j}(X)$ can be computed in polynomial time for every $i$ and $j$. Therefore, for every agent $i$ violating property (4), the outer while loop runs in polynomial time. Since there are at most $n$ agents violating property (4) and each execution of the outer while loop will fix one of them, the algorithm terminates in polynomial time with the cut configurations given as input. 
\end{proof}

\subsection{Satisfying Property (5)}\label{sec:allocations-proof3}

We now prove that our desired allocation satisfies property (5). Assume that after satisfying properties (1)-(4) there exist a non-envied agent $i \in [n]$ who violates property (5), i.e., $v_i(U_i(X)) > v_i(X_i)$. For every envied agent $j \in S$ adjacent to $i$ such that $A_{j,i}(X) \neq \emptyset$, we allocate this set to agent $i$ and remove the items from $E(i, j)$ that were previously allocated to $i$. Algorithm \ref{alg:non-envied-unalloc} and Lemma \ref{lem:prop-5} formalize this step.

\begin{algorithm}
\caption{Increase Non-Envied Vertices Valuation: Properties (1)-(4) + Property (5)}\label{alg:non-envied-unalloc}
\KwIn{A partial orientation $X$ that is an output of Algorithm~\ref{alg:prop-4} on a bipartite multi-graph $G$, satisfying properties $(1)$-$(4)$}
\KwOut{A partial orientation $X$ satisfying properties $(1)$-$(5)$}

\BlankLine

\While{there exists a non-envied vertex $i \in [n]$ such that $v_i(U_i(X)) > v_i(X_i)$}{

    \For{every envied vertex $j \in S$ adjacent to $i$ such that $A_{j, i}(X) \neq \emptyset$}{

        Let $D$ be the the set of items currently in $A_{j, i}(X)$ 

        $X_i \gets X_i \setminus E(i, j)$

        $X_i \gets X_i \cup D$
    }

    \If{property (4) is violated}{

        Run Algorithm \ref{alg:prop-4}
    
    }

}
 \Return{$X = (X_1, \dots, X_n)$}

\end{algorithm}

\begin{lemma}\label{lem:prop-5}
    Algorithm~\ref{alg:non-envied-unalloc} terminates and its output allocation satisfies properties (1)-(5). Moreover, it runs in polynomial time if the cut configurations are given as input.
\end{lemma}

\begin{proof}
    Let $X$ and $X'$ denote the input and output orientations of Algorithm \ref{alg:non-envied-unalloc} respectively. Consider a non-envied agent $i$ violating property (5) in $X$, i.e., $v_i(U_i(X)) > v_i(X_i)$. By Claim \ref{claim:unallocated_edges}, every unallocated edge incident to $i$ belongs to $A_{j, i}(X)$ for some envied vertex $j$ such that $A_{j, i}(X) \neq \emptyset$ and $E(i, j) \setminus A_{j, i}(X)$ is allocated to $i$. Therefore, Algorithm \ref{alg:non-envied-unalloc} allocates the set $U_i(X)$ to agent $i$. Figure \ref{fig:prop-5} demonstrates the unallocated edges incident to agent $i$. Observe that agent $i$ will not strongly envy any agent in $X'$. Let $j$ be an envied agent incident to $i$ with $A_{j, i}(X) \neq \emptyset$. Now, $j$ will not envy $i$ since we have,
    \begin{align*}
        v_j(X'_i) = v_j(A_{j, i}(X)) \le v_j(X_j),
    \end{align*}
    where the last inequality follows from property (3) of $X$. Hence, $X'$ remains $\efx$. Moreover, $X'$ is an orientation, and hence property (1) remains satisfied. Since agent $i$ was non-envied in $X$ and remains non-envied in $X'$, property (3) is also clearly satisfied. One can easily observe that property (2) continues to be satisfied, by construction. Therefore, after allocating the set $U_i(X)$ to agent $i$ using the for-loop in Algorithm \ref{alg:non-envied-unalloc}, the only property that might be violated is property (4), as it might be the case that agent $i$ previously envied some agent $k \in [n]$ in $X$, but since she has been better off, she no longer envies agent $k$, and thus, agent $k$ has now become non-envied while $A_k(X) \neq \emptyset$. Therefore, we run Algorithm \ref{alg:prop-4} in these cases to ensure property (4) also remains satisfied.
    
    Since, in every iteration of the while loop of Algorithm \ref{alg:non-envied-unalloc}, the number of vertices violating property (5) decreases by one, the algorithm must eventually terminate. Hence, the output allocation $X'$ satisfies all the first five properties. Moreover, it is clear that for each non-envied agent violating property (5), the inner for-loop takes polynomial time. Therefore, since Algorithm \ref{alg:prop-4} runs in polynomial time, we have that Algorithm \ref{alg:non-envied-unalloc} terminates in polynomial time given the cut configurations as input.
\end{proof}

\begin{figure}
    \centering

    \begin{tikzpicture}
    \Vertex[label = $j$, x = 0, y = 2, color = white]{1}
    \Vertex[label = $i$, x = 2, y = 2, color = white]{2}
    \Vertex[label = $k$, x = 4, y = 2, color = white]{3}
    
    \Edge[bend=40, Direct](1)(2)
    \Edge[bend=80, Direct](1)(2)
    \Edge[bend=-40, style = dashed](1)(2)
    \Edge[bend=-80, style = dashed](1)(2)

    \Edge[bend=40, Direct](3)(2)
    \Edge[bend=80, Direct](3)(2)
    \Edge[bend=-40, style = dashed](3)(2)
    \Edge[bend=-80, style = dashed](3)(2)
    
    \end{tikzpicture}

    \caption{A demonstration of a non-envied agent $i$ violating property (5) after satisfying the first four properties. Agents $j$ and $k$ are envied vertices adjacent to $i$ such that $A_{j, i}(X) \neq \emptyset$ and $A_{k, i}(X) \neq \emptyset$. Dashed edges correspond to unallocated items.}
    \label{fig:prop-5}
\end{figure}

Note that Claim \ref{claim:unallocated_edges} still holds since properties (1)-(4) are satisfied.

\subsection{Satisfying Property (6)}\label{sec:safe-set}
We now finally focus on satisfying property (6) where for any envied agent $i \in [n]$ who is envied by $j$ must be such that $j \in S_i(X)$. 

Algorithm~\ref{alg:safe-set-monotone} begins by identifying a pair of agents $(i,j)$ in $X$ where $i$ is envied by $j$ and $j \notin S_i(X)$ and swapping the bundles they possess from the \emph{$j$-cut} configuration of the set $E(i, j)$. Then, it allocates the set $A_i(X)$ to agent $i$ as well. We will show (in Lemma \ref{lem:safe-set-monotone}) that the above procedure satisfies property (6) while maintaining our five previous properties.

\begin{algorithm}
\caption{Safe Set: Properties (1)-(5)+Property (6)}\label{alg:safe-set-monotone}
\KwIn{Allocation X satisfying properties $(1)$-$(5)$}
\KwOut{Allocation X satisfying properties $(1)$-$(6)$}

\BlankLine

\While{there exists an $i \in [n]$ who is envied by $j \notin S_i(X)$}{

    Let the partition $(C_1, C_2)$ be the \emph{$j$-cut} configuration of the set $E(i, j)$. 

    Swap the bundles $C_1$ and $C_2$ between agents $i$ and $j$.

    $X_i \gets X_i \cup A_i(X)$

    \If{property (4) is violated}{

        Run Algorithm \ref{alg:prop-4}
    
    }

    \If{property (5) is violated}{

        Run Algorithm \ref{alg:non-envied-unalloc}
    
    }

}
 \Return{$X = (X_1, \dots, X_n)$}
\end{algorithm}

\begin{lemma}\label{lem:safe-set-monotone}
    Algorithm~\ref{alg:safe-set-monotone} terminates and outputs a partial allocation that satisfies properties (1)-(6). Moreover, it runs in polynomial time if the cut configurations are given as input.
\end{lemma}

\begin{proof}
    Consider a single iteration of the while loop of Algorithm \ref{alg:safe-set-monotone}. Let $i \in [n]$ be the agent violating property (6) and $j$ be the agent who envies $i$. Let $X$ and $X'$ be the allocations before and after this iteration, respectively. We have,
    $$v_i (X_i') = v_i( X_j \cup A_i(X)) > v_i(X_i),$$ where the last inequality comes from the fact that agent $i$ was an agent violating property (6) previously. Note that, agent $i$ is now better off and does not envy any other agent. During the execution of Lines 2-4 (in Algorithm~\ref{alg:safe-set-monotone}), we swap the bundles $C_1$ and $C_2$ between agents $i$ and $j$, where $(C_1, C_2)$ is the \emph{$j$-cut} configuration of the set $E(i, j)$ and add $A_i(X)$ to $X_i$. Therefore, agents $i$ and $j$ are better off and are both non-envied, and the allocation remains $\efx$. Using a similar argument as in the proof of Lemma \ref{lem:prop-5}, we have that $X'$ satisfies properties (1)-(3).

    Since agent $j$ is now better off, it might be the case that agent $j$ previously envied some agent $k \in [n]$ in $X$, but since she has been better off, she no longer envies agent $k$, and thus, agent $k$ has now become non-envied while $A_k(X) \neq \emptyset$. Therefore, we run Algorithm \ref{alg:prop-4} in these cases to ensure property (4) also remains satisfied. Moreover, it might be the case that agent $k$ violates property (5) after executing line 6. Therefore, the second if statement runs Algorithm \ref{alg:non-envied-unalloc} to ensure property (5) is also satisfied. Note that executing Algorithm~\ref{alg:non-envied-unalloc} maintains the first five properties.
    
    Now, since each iteration of the algorithm decreases the number of envied vertices strictly, we know that Algorithm~\ref{alg:safe-set-monotone} terminates in at most $n$ iterations. Moreover, given the cut configurations as input, since Algorithms \ref{alg:prop-4} and \ref{alg:non-envied-unalloc} both terminate in polynomial time, Algorithm \ref{alg:safe-set-monotone} also terminates in polynomial time.
\end{proof}

Notice that Claim \ref{claim:unallocated_edges} still holds since properties (1)-(4) are satisfied.

\subsection{Allocating the Remaining Items} \label{sec:complete-allocation}

For a given bipartite multi-graph, we execute Algorithms~\ref{alg:greedy}, \ref{alg:prop-4}, \ref{alg:non-envied-unalloc}, and \ref{alg:safe-set-monotone} (in that order) and reach a desired partial $\efx$ orientation $X$ that satisfies properties (1)-(6). What is remaining is to now make $X$ complete by assigning the unallocated items while maintaining $\efx$ guarantees. We will show that the six properties of $X$ make it easy to do so. By Claim \ref{claim:unallocated_edges}, we know that the only unallocated edges in $X$ are the set $A_i(X)$ for some envied vertex $i \in [n]$. Let $j \in [n]$ be the agent who envies $i$. We will allocate $A_i(X)$ to agent $j$. The procedure is formalized in Algorithm \ref{alg:final}. We will prove that this algorithm outputs an $\efx$ allocation.

\begin{algorithm}
\caption{Finalize Allocation}\label{alg:final}
\KwIn{Allocation X satisfying properties $(1)$-$(6)$}
\KwOut{An $\efx$ allocation}

\While{there exists an envied agent $i \in [n]$ such that $A_i (X) \neq \emptyset$}{

    Let $j$ be the agent who envies $i$.

    $X_j \gets X_j \cup A_i(X)$
    
}
 \Return{$X = (X_1, \dots, X_n)$}
\end{algorithm}

\begin{lemma}\label{lem:final}
    Algorithm~\ref{alg:final} terminates and outputs a complete $\efx$ allocation. Moreover, it runs in polynomial time if the cut configurations are given as input.
\end{lemma}

\begin{proof}
    To begin with, note that Algorithm~\ref{alg:final} terminates in polynomial time if the cut configurations are given as input, since in each iteration at least one edge will be allocated. Moreover, by Claim \ref{claim:unallocated_edges}, it is clear Algorithm \ref{alg:final} outputs a complete allocation $X$. Hence, we only need to prove that $X$ is $\efx$. 

    Consider an agent $j \in [n]$ who envies agents $i_1, \cdots, i_l$. After the execution of Algorithm \ref{alg:final}, this agent has received the set $X_j \cup A_{i_1}(X) \cup \cdots \cup A_{i_l}(X)$. It suffices to show that no agent envies agent $j$. Observe that $i_1, \cdots, i_l \in S$ and $j \in T$, and also each set $A_{i_p}(X)$ is only valuable to agent $i_p \in S$ and some other agent, called $j_p \in T$. Now, consider an arbitrary agent $k$. If for every $p \in [l]$, $k \neq j_p$ and $k \neq i_p$, then we have that the set $A_{i_1}(X) \cup \cdots \cup A_{i_l}(X)$ is not valuable for agent $k$. Since agent $j$ was previously non-envied, we have $v_k(X_j \cup A_{i_1}(X) \cup \cdots \cup A_{i_l}(X)) = v_k(X_j) \le v_k(X_k)$, meaning that $k$ does not envy $j$. Now assume $k = j_p$ for some $p \in [l]$. Observe that $$v_k(A_{i_1}(X) \cup \cdots \cup A_{i_l}(X)) = v_k(A_{i_p}(X)) \le v_k(U_k(X)),$$ where the last inequality comes from the fact that $A_{i_p}(X) \subseteq U_k(X)$. Moreover, note that $k \in T$ and since the allocation prior to the execution of Algorithm \ref{alg:final} was an orientation, $j$ has not received any item valuable to agent $k$. Thus, $$v_k(X_j \cup A_{i_1}(X) \cup \cdots \cup A_{i_l}(X)) = v_k(A_{i_p}(X)) \le v_k(U_k(X)) \le v_k(X_k),$$ where the last inequality comes from property (5). Now assume $k = i_p$ for some $p \in [l]$. We have $$v_k(X_j \cup A_{i_1}(X) \cup \cdots \cup A_{i_l}(X)) = v_k(X_j \cup A_{i_p}(X))\le v_k(X_k),$$ where the last inequality comes from property (6). Therefore, in either case, agent $k$ will not envy agent $j$ and the allocation remains $\efx$.
\end{proof}

Now, we are ready to prove our main result. 

\begin{theorem} \label{thm:multi-bipartite-main}
    For any fair division instance on a bipartite multi-graph with monotone valuations, $\efx$ allocations always exist and can be computed in pseudo-polynomial time. Moreover, if the valuations are cancelable, such allocations can be computed in polynomial time.
\end{theorem}

\begin{proof}
    We run Algorithms \ref{alg:greedy}, \ref{alg:prop-4}, \ref{alg:non-envied-unalloc}, \ref{alg:safe-set-monotone}, and \ref{alg:final} in the mentioned order to obtain our $\efx$ allocation. Regarding the running time, by Lemmas \ref{lem:greedy}, \ref{lem:runtime-prop(4)}, \ref{lem:prop-5}, \ref{lem:safe-set-monotone}, and \ref{lem:final}, we can compute our allocation in polynomial time if the configurations are given as input. Note that, using Remark \ref{rem:conf-runtime}, we can compute the cut configurations in polynomial time for cancelable valuations and in pseudo-polynomial time for monotone valuations, which completes our proof.
\end{proof}

Now, as a corollary, we obtain the following theorem that says that we can compute orientations that are $1/2$-$\efx$ in polynomial time for additive valuations.

\begin{restatable}{theorem}{approxefxthm}\label{thm:1/2-efx}
    For any fair division instance on a bipartite multi-graph with additive valuations, there always exists an orientation where at least $\lceil\frac{n}{2}\rceil$ of agents are $\efx$ and the remaining agents are $\frac{1}{2}$-$\efx$. Furthermore, such an orientation can be computed in polynomial time.
\end{restatable}
\begin{proof}
    For a given instance, let $S = \{i_1, i_2, \ldots, i_{|S|}\}$ and $T = \{j_1, j_2, \ldots, j_{|T|}\}$ be two parts of the given bipartite multi-graph. Without loss of generality, we assume assume $|S| \leq |T|$. Let $X$ be the partial orientation obtained after executing Algorithms~\ref{alg:greedy} and \ref{alg:prop-4} (in that order). Note that the allocation $X$ satisfies properties (1)-(4).

     Using Claim \ref{claim:unallocated_edges}, we know that the only unallocated bundles in $X$ are between an envied agent $i$ and a non-envied agent $j$, and the set $E(i,j) \setminus A_{i,j}(X)$ is allocated to $j$. So, let us consider such a pair of agents $(i,j)$ with an unallocated set $C$ adjacent to them. Note that after satisfying property (3), we have $v_i(C) \leq v_i(X_i)$. Since $i$ does not envy $j$ we also have $v_i(X_j) \leq v_i(X_i)$. Combining the last inequalities, we have 
     \begin{align}
         v_i(X_i) \geq \frac{1}{2} v_i(X_j \cup C)
     \end{align}
     We now allocate $C$ to $j$, and by above inequality, agent $i$ is at least $\frac{1}{2}$-$\efx$ with $j$'s new bundle. 
     
    We do this for every such envied and non-envied pair of agents and obtain a full orientation (say, $Y$) for the given multi-graph. Therefore, every agent in $T$ is $\efx$-satisfied in $Y$ while, agents in  $S$ are $\frac{1}{2}$-$\efx$ satisfied in $Y$. Since $|S| \leq |T|$, the stated claim stands proven. Finally, using Remark \ref{rem:conf-runtime}, the cut configurations can be computed in polynomial time for additive valuations, and hence polynomial-time computability of such orientations follow.
\end{proof}
\section{Further Improvements and Limitations}
\label{sec:improve}

Theorem \ref{thm:multi-bipartite-main} motivates the question of what happens if the graph skeleton contains cycles of odd length. In this section, we extend our techniques to prove that any multi-cycle instance with $\mms$-feasible valuations admits an $\efx$ allocation (Theorem~\ref{thm:multi-cycle}). This demonstrates the power of our technique while providing insight and strong hope for potentially proving the existence of $\efx$ on general multi-graphs.

\begin{restatable}{theorem}{multicycleefxthm}\label{thm:multi-cycle}
     For any fair division instance on multi-cycles with $\mms$-feasible valuations, $\efx$ allocations always exist and can be computed in pseudo-polynomial time.
\end{restatable}
\begin{proof}
    To begin with, note that multi-cycles with even length are bipartite graphs, for which Theorem~\ref{thm:multi-bipartite-main} proves the existence of $\efx$ allocations. Therefore, we only consider odd-length multi-cycles with at least five vertices.\footnote{Existence of $\efx$ allocation on odd-length multi-cycles with three vertices is known \cite{akrami2025efx} for $\mms$-feasible valuations.} To prove the stated theorem, we distinguish between two main cases, which are further divided into different sub-cases. Here, for every pair of adjacent agents $i$ and $j$, we define both the \emph{$i$-cut} and \emph{$j$-cut} configurations of the set $E(i, j)$ in the same manner as before (see Definition~\ref{def:configuration}).

    \begin{itemize}
        \item \textbf{Case 1. There exists a pair of adjacent agents $i, j \in [n]$ such that they prefer different bundles in one of the configurations of $E(i,j)$.} Without loss of generality, we assume $(C_1, C_2)$ is the \emph{$i$-cut} configuration of $E(i, j)$ where $v_i(C_1) \ge v_i(C_2)$ and $v_j(C_2) \ge v_j(C_1)$. In this case, we remove the items of $E(i, j)$ from the graph, which makes the remaining graph a multi-$P_n$. That is, it becomes a bipartite multi-graph. Since the length of this path is even, we can define the sets $S$ and $T$ of the vertices in the bipartite graph so that both $i$ and $j$ belong to the set $T$. Let $X$ be the $\efx$ allocation of this multi-bipartite graph using the sets $S$ and $T$ as described. Therefore, using Lemma~\ref{lem:envy-S}, we know that agents $i$ and $j$ are both non-envied agents in $X$. Now, we add the edges $E(i,j)$ and allocate $C_1$ to $i$ and $C_2$ to $j$. First, note that agents $i$ and $j$ do not envy each other based on the assumption of our case. Furthermore, they remain non-envied, and the resulting allocation is $\efx$.
        
        \item \textbf{Case 2. For every pair of adjacent agents $i, j \in [n]$, they prefer the same bundle in both configurations of the set $E(i,j)$.} Let us assume $j', j, i, i'$ to be four consecutive agents in the cycle. We remove the agents $i$ and $j$ and their incident edges from the graph. Then, the remaining graph is a multi-$P_{n-2}$ with even length. Therefore, following the same argument as in the previous case, allocation $X$ exists in the remaining graph that is $\efx$, and both agents $j'$ and $i'$ are non-envied. 
        Let $(C_1, C_2)$, $(D_1, D_2)$, and $(E_1, E_2)$ be the \emph{$j'$-cut}, \emph{$i$-cut}, and \emph{$i'$-cut} for the sets $E(j', j)$, $E(j, i)$, and $E(i, i')$ respectively. Without loss of generality, by the case definition, we assume the following:
        \begin{align*}
            v_{j'}(C_1) \ge v_{j'}(C_2), \text{ and } v_{j}(C_1) \ge v_{j}(C_2) \\
            v_{j}(D_1) \ge v_{j}(D_2), \text{ and } v_{i}(D_1) \ge v_{i}(D_2) \\
            v_{i}(E_1) \ge v_{i}(E_2), \text{ and } v_{i'}(E_1) \ge v_{i'}(E_2)
        \end{align*}
        We now consider the following cases. (In each case, we allocate the sets $C_1, C_2, D_1, D_2, E_1, E_2$ to agents $j', j, i, i'$. Note that for each case, one can easily check that agents $j', j, i, i'$ do not strongly envy each other. Furthermore, agents $j'$ and $i'$ remain non-envied.)

        \begin{itemize}
            \item \textbf{Case 2.1.1. $v_j(C_2 \cup D_2) \ge max(v_j(C_1), v_j(D_1))$ and $v_i(D_1 \cup E_2) \ge v_i(E_1)$:} 

            We allocate $C_1$ to $j'$, $C_2 \cup D_2$ to $j$, $D_1 \cup E_2$ to $i$, and $E_1$ to $i'$. 
            
            \item \textbf{Case 2.1.2. $v_j(C_2 \cup D_2) \ge max(v_j(C_1), v_j(D_1))$ and $v_i(D_1 \cup E_2) < v_i(E_1)$:}

            We allocate $C_1$ to $j'$, $C_2 \cup D_2$ to $j$, $E_1$ to $i$, and $D_1 \cup E_2$ to $i'$.
            
            \item \textbf{Case 2.2.1. $v_j(C_1) \ge max(v_j(C_2 \cup D_2), v_j(D_1))$ and $v_i(D_1 \cup E_2) \ge v_i(E_1)$:}

            We allocate $C_2 \cup D_2$ to $j'$, $C_1$ to $j$, $D_1 \cup E_2$ to $i$, and $E_1$ to $i'$.
            
            \item \textbf{Case 2.2.2. $v_j(C_1) \ge max(v_j(C_2 \cup D_2), v_j(D_1))$ and $v_i(D_1 \cup E_2) < v_i(E_1)$:}

            We allocate $C_2 \cup D_2$ to $j'$, $C_1$ to $j$, $E_1$ to $i$, and $D_1 \cup E_2$ to $i'$.
            
            \item \textbf{Case 2.3.1. $v_j(D_1) \ge max(v_j(C_2 \cup D_2), v_j(C_1))$ and $v_i(D_2 \cup E_2) \ge v_i(E_1)$:}

            We allocate $C_1$ to $j'$, $D_1$ to $j$, $D_2 \cup E_2$ to $i$, and $C_2 \cup E_1$ to $i'$.
            
            \item \textbf{Case 2.3.2. $v_j(D_1) \ge max(v_j(C_2 \cup D_2), v_j(C_1))$ and $v_i(D_2 \cup E_2) < v_i(E_1)$:}  

            We allocate $C_1$ to $j'$, $D_1$ to $j$, $E_1$ to $i$, and $C_2 \cup D_2 \cup E_2$ to $i'$.
        \end{itemize}
    \end{itemize}

Therefore, we can find an $\efx$ allocation in multi-cycles of odd length, thereby completing our proof. The running time of this procedure can be concluded from the running time of our algorithm for bipartite multi-graphs (Theorem \ref{thm:multi-bipartite-main}).
\end{proof}

Note that our proof works for monotone valuations if the length of the multi-cycle is at least four. Moreover, it is clear to see that it terminates in polynomial time for cancelable valuations. Therefore, we can conclude the following:

\begin{corollary}\label{cor:cycle}
         For any fair division instance on multi-cycles with at least four agents and monotone valuations, $\efx$ allocations always exist and can be computed in pseudo-polynomial time. Moreover, if the valuations are cancelable, such allocations can be computed in polynomial time.
\end{corollary}

\paragraph{Why do our techniques fail for general multi-graphs?}
Theorems \ref{thm:multi-bipartite-main} and \ref{thm:multi-cycle} make it hopeful that one can adapt/modify our techniques to prove the existence of $\efx$ allocations in general multi-graphs. In our proof, we never had to deal with the case (after Algorithm~\ref{alg:greedy}) where there were two envied vertices $i$ and $j$ in the graph such that no edge from $E(i,j)$ was allocated. This turns out to be the most complicated case for the general multi-graph structure. As a solution concept, one can aim to achieve a partial orientation with $|S_i(X) \cap S_j(X)| \ge 2$ for any adjacent envied vertices $i,j$. If this happens, we can let $(C_1, C_2)$ be the \emph{$i$-cut} configuration of items $E(i, j)$ and then allocate $C_1$ and $C_2$ to two different vertices in $S_i(X) \cap S_j(X)$. Also, we believe that we have to use both configurations \emph{$i$-cut} and \emph{$j$-cut} for every pair of adjacent vertices $i,j$ for extending our result to general multi-graphs. 

\vspace{6pt}
\begin{mdframed}
    \textit{Conjecture:} Any fair division instance on a multi-graph admits an $\efx$ allocation.
\end{mdframed}

% \MA{Furthermore, we believe that by putting some constraints on the graph structure, one can increase the number of non-envied vertices in safe sets and potentially satisfy the condition $|S_i(X) \cap S_j(X)| \ge 2$ for every pair of adjacent and envied vertices like $i$ and $j$.}

\section{Conclusion}

In this work, we study a model that captures the setting where every item is relevant to at most two agents and any two agents can have multiple relevant items in common (represented by multi-graphs). We prove the existence of $\efx$ allocations for fair division instances on bipartite multi-graphs with monotone valuations and multi-cycles with $\mms$ feasible valuations, and that they can be computed in pseudo-polynomial time. Moreover, for bipartite multi-graphs and multi-cycles with length at least four, they can be computed in polynomial time for cancelable valuation functions. An immediate question for future research work is to understand $\efx$ allocations on general multi-graphs, as discussed in Section~\ref{sec:improve}. 

The non-existence of $\efx$ orientations on bipartite multi-graphs implies that some degree of \emph{wastefulness} is inherent in $\efx$ allocations for such instances. In this work, we prove the existence of orientations that are $\efx$ for half the agents and $1/2$-$\efx$ for the remaining agents for instances represented via a bipartite multi-graph with additive valuations. Another immediate question is, therefore, to improve this factor. Finally, it would be interesting to see what one can say about approximating social welfare or Nash social welfare of $\efx$ allocations to understanding the trade-offs with fairness and efficiencies in multi-graph instances.

Ultimately, we hope that the insights gained from exploring $\efx$ allocations in the multi-graph setting will contribute to advancements in the broader challenge of understanding $\efx$ allocations, in general.
\section*{Acknowledgments}
N.\ R.\ is supported by the Lise Meitner Postdoctoral Fellowship (2023-25), by the Max Planck Institute for Informatics, SIC, Germany. A part of the work was done when M.\ A.\ and M.\ K.\ were interns at Max Planck Institute for Informatics, SIC, Germany.

\bibliographystyle{abbrvnat}
\bibliography{ref}

\begin{thebibliography}{48}
\providecommand{\natexlab}[1]{#1}
\providecommand{\url}[1]{\texttt{#1}}
\expandafter\ifx\csname urlstyle\endcsname\relax
  \providecommand{\doi}[1]{doi: #1}\else
  \providecommand{\doi}{doi: \begingroup \urlstyle{rm}\Url}\fi

\bibitem[Afshinmehr et~al.(2024)Afshinmehr, Ansaripour, Danaei, and
  Mehlhorn]{afshinmehr2024approximateefxexacttefx}
M.~Afshinmehr, M.~Ansaripour, A.~Danaei, and K.~Mehlhorn.
\newblock Approximate {EFX} and exact {tEFX} allocations for indivisible
  chores: Improved algorithms, 2024.
\newblock URL \url{https://arxiv.org/abs/2410.18655}.

\bibitem[Akrami and Rathi(2025)]{akrami2025epistemic}
H.~Akrami and N.~Rathi.
\newblock Epistemic {EFX} allocations exist for monotone valuations.
\newblock In \emph{Proceedings of the AAAI Conference on Artificial
  Intelligence}, volume~39, pages 13520--13528, 2025.

\bibitem[Akrami et~al.(2022)Akrami, Rezvan, and Seddighin]{ef2x}
H.~Akrami, R.~Rezvan, and M.~Seddighin.
\newblock An {EF2X} allocation protocol for restricted additive valuations.
\newblock In \emph{Proceedings of the Thirty-First International Joint
  Conference on Artificial Intelligence, {IJCAI}}, pages 17--23, 2022.

\bibitem[Akrami et~al.(2025)Akrami, Alon, Chaudhury, Garg, Mehlhorn, and
  Mehta]{akrami2025efx}
H.~Akrami, N.~Alon, B.~R. Chaudhury, J.~Garg, K.~Mehlhorn, and R.~Mehta.
\newblock {EFX}: a simpler approach and an (almost) optimal guarantee via
  rainbow cycle number.
\newblock \emph{Operations Research}, 73\penalty0 (2):\penalty0 738--751, 2025.

\bibitem[Amanatidis et~al.(2020)Amanatidis, Markakis, and
  Ntokos]{amanatidis2020multiple}
G.~Amanatidis, E.~Markakis, and A.~Ntokos.
\newblock Multiple birds with one stone: Beating 1/2 for {EFX} and {GMMS} via
  envy cycle elimination.
\newblock \emph{Theoretical Computer Science}, 841:\penalty0 94--109, 2020.

\bibitem[Amanatidis et~al.(2021)Amanatidis, Birmpas, Filos-Ratsikas, Hollender,
  and Voudouris]{amanatidis2021maximum}
G.~Amanatidis, G.~Birmpas, A.~Filos-Ratsikas, A.~Hollender, and A.~A.
  Voudouris.
\newblock Maximum {N}ash welfare and other stories about {{EFX}}.
\newblock \emph{Journal of Theoretical Computer Science}, 863:\penalty0 69--85,
  2021.

\bibitem[Amanatidis et~al.(2023)Amanatidis, Aziz, Birmpas, Filos-Ratsikas, Li,
  Moulin, Voudouris, and Wu]{survey2023}
G.~Amanatidis, H.~Aziz, G.~Birmpas, A.~Filos-Ratsikas, B.~Li, H.~Moulin, A.~A.
  Voudouris, and X.~Wu.
\newblock Fair division of indivisible goods: Recent progress and open
  questions.
\newblock \emph{Artificial Intelligence}, 322:\penalty0 103965, 2023.
\newblock ISSN 0004-3702.

\bibitem[Ashuri et~al.(2025)Ashuri, Gkatzelis, and Sgouritsa]{AGS25}
A.~Ashuri, V.~Gkatzelis, and A.~Sgouritsa.
\newblock {EF2X} exists for four agents.
\newblock In \emph{AAAI-25, Sponsored by the Association for the Advancement of
  Artificial Intelligence, February 25 - March 4, 2025, Philadelphia, PA,
  {USA}}. {AAAI} Press, 2025.

\bibitem[Berendsohn et~al.(2022)Berendsohn, Boyadzhiyska, and Kozma]{aram22}
B.~A. Berendsohn, S.~Boyadzhiyska, and L.~Kozma.
\newblock Fixed-point cycles and approximate {{EFX}} allocations.
\newblock In S.~Szeider, R.~Ganian, and A.~Silva, editors, \emph{47th
  International Symposium on Mathematical Foundations of Computer Science,
  {MFCS} 2022, August 22-26, 2022, Vienna, Austria}, volume 241 of
  \emph{LIPIcs}, pages 17:1--17:13. Schloss Dagstuhl - Leibniz-Zentrum
  f{\"{u}}r Informatik, 2022.

\bibitem[Berger et~al.(2022)Berger, Cohen, Feldman, and Fiat]{berger2021almost}
B.~Berger, A.~Cohen, M.~Feldman, and A.~Fiat.
\newblock Almost full {{EFX}} exists for four agents.
\newblock In \emph{Proceedings of the 36th AAAI Conference on Artificial
  Intelligence ({AAAI})}, volume 36(5), pages 4826--4833, 2022.

\bibitem[Bhaskar and Pandit(2024)]{BP24}
U.~Bhaskar and Y.~Pandit.
\newblock {EFX} allocations on some multi-graph classes.
\newblock \emph{CoRR}, abs/2412.06513, 2024.

\bibitem[Brams and Taylor(1996)]{brams1996fair}
S.~J. Brams and A.~D. Taylor.
\newblock \emph{Fair Division: From cake-cutting to dispute resolution}.
\newblock Cambridge University Press, 1996.

\bibitem[Brandt and Procaccia(2016)]{brandt2016handbook}
F.~Brandt and A.~D. Procaccia.
\newblock \emph{Handbook of computational social choice}.
\newblock Cambridge University Press, 2016.

\bibitem[Budish(2011)]{budish2011combinatorial}
E.~Budish.
\newblock The combinatorial assignment problem: Approximate competitive
  equilibrium from equal incomes.
\newblock \emph{Journal of Political Economy}, 119\penalty0 (6):\penalty0
  1061--1103, 2011.

\bibitem[Budish and Cantillon(2012)]{budish2012multi}
E.~Budish and E.~Cantillon.
\newblock The multi-unit assignment problem: Theory and evidence from course
  allocation at harvard.
\newblock \emph{American Economic Review}, 102\penalty0 (5):\penalty0
  2237--2271, 2012.

\bibitem[Caragiannis et~al.(2016)Caragiannis, Kurokawa, Moulin, Procaccia,
  Shah, and Wang]{caragiannis2016unreasonable}
I.~Caragiannis, D.~Kurokawa, H.~Moulin, A.~D. Procaccia, N.~Shah, and J.~Wang.
\newblock The unreasonable fairness of maximum {N}ash welfare.
\newblock In \emph{Proceedings of the 2016 ACM Conference on Economics and
  Computation}, pages 305--322. ACM, 2016.

\bibitem[Caragiannis et~al.(2019)Caragiannis, Gravin, and
  Huang]{caragiannis2019envy}
I.~Caragiannis, N.~Gravin, and X.~Huang.
\newblock Envy-freeness up to any item with high {N}ash welfare: The virtue of
  donating items.
\newblock In \emph{Proceedings of the 20th ACM Conference on Economics and
  Computation ({EC})}, pages 527--545, 2019.

\bibitem[Caragiannis et~al.(2023)Caragiannis, Garg, Rathi, Sharma, and
  Varricchio]{Caragiannis2023}
I.~Caragiannis, J.~Garg, N.~Rathi, E.~Sharma, and G.~Varricchio.
\newblock New fairness concepts for allocating indivisible items.
\newblock In E.~Elkind, editor, \emph{Proceedings of the Thirty-Second
  International Joint Conference on Artificial Intelligence, {IJCAI-23}}, pages
  2554--2562. International Joint Conferences on Artificial Intelligence
  Organization, 2023.

\bibitem[Chan et~al.(2019)Chan, Chen, Li, and Wu]{chan2019maximin}
H.~Chan, J.~Chen, B.~Li, and X.~Wu.
\newblock Maximin-aware allocations of indivisible goods.
\newblock In \emph{Proceedings of the 28th International Joint Conference on
  Artificial Intelligence ({IJCAI})}, pages 137--143, 2019.

\bibitem[Chaudhury et~al.(2020)Chaudhury, Garg, and Mehlhorn]{chaudhury2020efx}
B.~R. Chaudhury, J.~Garg, and K.~Mehlhorn.
\newblock {EFX} exists for three agents.
\newblock In \emph{Proceedings of the 21st ACM Conference on Economics and
  Computation ({EC})}, pages 1--19, 2020.

\bibitem[Chaudhury et~al.(2021)Chaudhury, Kavitha, Mehlhorn, and
  Sgouritsa]{chaudhury2021little}
B.~R. Chaudhury, T.~Kavitha, K.~Mehlhorn, and A.~Sgouritsa.
\newblock A little charity guarantees almost envy-freeness.
\newblock \emph{SIAM Journal on Computing}, 50\penalty0 (4):\penalty0
  1336--1358, 2021.

\bibitem[Christodoulou et~al.(2023)Christodoulou, Fiat, Koutsoupias, and
  Sgouritsa]{christodoulou2023fair}
G.~Christodoulou, A.~Fiat, E.~Koutsoupias, and A.~Sgouritsa.
\newblock Fair allocation in graphs.
\newblock In \emph{Proceedings of the 24th ACM Conference on Economics and
  Computation ({EC})}, pages 473--488, 2023.

\bibitem[Christoforidis and Santorinaios(2024)]{CS24}
V.~Christoforidis and C.~Santorinaios.
\newblock On the pursuit of {EFX} for chores: Non-existence and approximations.
\newblock In \emph{Proceedings of the 33rd International Joint Conference on
  Artificial Intelligence, {(IJCAI)}}, pages 2713--2721, 2024.

\bibitem[Deligkas et~al.(2024)Deligkas, Eiben, Goldsmith, and
  Korchemna]{deligkas2024ef1}
A.~Deligkas, E.~Eiben, T.-L. Goldsmith, and V.~Korchemna.
\newblock {EF1} and {EFX} orientations.
\newblock \emph{arXiv preprint arXiv:2409.13616}, 2024.

\bibitem[Dubins and Spanier(1961)]{dubins1961cut}
L.~E. Dubins and E.~H. Spanier.
\newblock How to cut a cake fairly.
\newblock \emph{The American Mathematical Monthly}, 68\penalty0 (1P1):\penalty0
  1--17, 1961.

\bibitem[Etkin et~al.(2007)Etkin, Parekh, and Tse]{etkin2007spectrum}
R.~Etkin, A.~Parekh, and D.~Tse.
\newblock Spectrum sharing for unlicensed bands.
\newblock \emph{IEEE Journal on Selected Areas in Communications}, 25\penalty0
  (3):\penalty0 517--528, 2007.

\bibitem[Farhadi et~al.(2021)Farhadi, Hajiaghayi, Latifian, Seddighin, and
  Yami]{farhadi2021almost}
A.~Farhadi, M.~Hajiaghayi, M.~Latifian, M.~Seddighin, and H.~Yami.
\newblock Almost envy-freeness, envy-rank, and nash social welfare matchings.
\newblock In \emph{Proceedings of the AAAI Conference on Artificial
  Intelligence}, volume~35, pages 5355--5362, 2021.

\bibitem[Foley(1967)]{foley1967resource}
D.~K. Foley.
\newblock Resource allocation and the public sector.
\newblock 1967.

\bibitem[Garg and Murhekar(2025)]{garg2025existence2efxallocationschores}
J.~Garg and A.~Murhekar.
\newblock Existence of 2-{EFX} allocations of chores, 2025.
\newblock URL \url{https://arxiv.org/abs/2507.19461}.

\bibitem[Garg et~al.(2025)Garg, Murhekar, and Qin]{GMQ25}
J.~Garg, A.~Murhekar, and J.~Qin.
\newblock Constant-factor {EFX} exists for chores.
\newblock In \emph{Proceedings of the 57th ACM Symposium on Theory of Computing
  (STOC)}, 2025.
\newblock ISBN 9798400715105.

\bibitem[Halpern et~al.(2020)Halpern, Procaccia, Psomas, and
  Shah]{halpern2020fair}
D.~Halpern, A.~D. Procaccia, A.~Psomas, and N.~Shah.
\newblock Fair division with binary valuations: One rule to rule them all.
\newblock In \emph{Web and Internet Economics: 16th International Conference,
  WINE 2020, Beijing, China, December 7--11, 2020, Proceedings 16}, pages
  370--383. Springer, 2020.

\bibitem[Jahan et~al.(2023)Jahan, Seddighin, Javadi, and Sharifi]{chasmjahan23}
S.~C. Jahan, M.~Seddighin, S.~M.~S. Javadi, and M.~Sharifi.
\newblock Rainbow cycle number and {{EFX}} allocations: (almost) closing the
  gap.
\newblock In \emph{Proceedings of the Thirty-Second International Joint
  Conference on Artificial Intelligence, {IJCAI} 2023, 19th-25th August 2023,
  Macao, SAR, China}, pages 2572--2580. ijcai.org, 2023.

\bibitem[Lipton et~al.(2004)Lipton, Markakis, Mossel, and
  Saberi]{lipton2004approximately}
R.~J. Lipton, E.~Markakis, E.~Mossel, and A.~Saberi.
\newblock On approximately fair allocations of indivisible goods.
\newblock In \emph{Proceedings of the 5th ACM conference on Electronic
  commerce}, pages 125--131. ACM, 2004.

\bibitem[Mahara(2021)]{mahara2021extension}
R.~Mahara.
\newblock Extension of additive valuations to general valuations on the
  existence of {{EFX}}.
\newblock In \emph{Proceedings of 29th Annual European Symposium on Algorithms
  ({ESA})}, page 66:1–66:15, 2021.
\newblock \doi{10.4230/LIPIcs.ESA.2021.66}.

\bibitem[Mahara(2025)]{mahara2025existencefairefficientallocation}
R.~Mahara.
\newblock Existence of fair and efficient allocation of indivisible chores,
  2025.
\newblock URL \url{https://arxiv.org/abs/2507.09544}.

\bibitem[Moulin(2004)]{moulin2004fair}
H.~Moulin.
\newblock \emph{Fair division and collective welfare}.
\newblock MIT press, 2004.

\bibitem[Plaut and Roughgarden(2020)]{plaut2020almost}
B.~Plaut and T.~Roughgarden.
\newblock Almost envy-freeness with general valuations.
\newblock \emph{SIAM Journal on Discrete Mathematics}, 34\penalty0
  (2):\penalty0 1039--1068, 2020.

\bibitem[Pratt and Zeckhauser(1990)]{pratt1990fair}
J.~W. Pratt and R.~J. Zeckhauser.
\newblock The fair and efficient division of the winsor family silver.
\newblock \emph{Management Science}, 36\penalty0 (11):\penalty0 1293--1301,
  1990.

\bibitem[Procaccia(2020)]{procaccia2020technical}
A.~D. Procaccia.
\newblock Technical perspective: An answer to fair division's most enigmatic
  question.
\newblock \emph{Communications of the ACM}, 63\penalty0 (4):\penalty0 118--118,
  2020.

\bibitem[Robertson and Webb(1998)]{robertson1998cake}
J.~Robertson and W.~Webb.
\newblock \emph{Cake-cutting algorithms: Be fair if you can}.
\newblock AK Peters/CRC Press, 1998.

\bibitem[Sgouritsa and Sotiriou(2025)]{SS25}
A.~Sgouritsa and M.~M. Sotiriou.
\newblock On the existence of {EFX} allocations in multigraphs.
\newblock In \emph{Proceedings of the 24th International Conference on
  Autonomous Agents and Multiagent Systems, {AAMAS} 2025, Detroit, MI, USA, May
  19-23, 2025}. International Foundation for Autonomous Agents and Multiagent
  Systems / {ACM}, 2025.

\bibitem[Steinhaus(1948)]{steinhaus1948problem}
H.~Steinhaus.
\newblock The problem of fair division.
\newblock \emph{Econometrica}, 16:\penalty0 101--104, 1948.

\bibitem[Stromquist(1980)]{stromquist1980cut}
W.~Stromquist.
\newblock How to cut a cake fairly.
\newblock \emph{The American Mathematical Monthly}, 87\penalty0 (8):\penalty0
  640--644, 1980.

\bibitem[Su(1999)]{edward1999rental}
F.~E. Su.
\newblock Rental harmony: Sperner's lemma in fair division.
\newblock \emph{The American mathematical monthly}, 106\penalty0 (10):\penalty0
  930--942, 1999.

\bibitem[Vossen(2002)]{vossen2002fair}
T.~Vossen.
\newblock \emph{Fair allocation concepts in air traffic management}.
\newblock PhD thesis, University of Maryland, College Park, 2002.

\bibitem[Zeng and Mehta(2024)]{zeng2024structure}
J.~A. Zeng and R.~Mehta.
\newblock On the structure of envy-free orientations on graphs.
\newblock \emph{arXiv preprint arXiv:2404.13527}, 2024.

\bibitem[Zhou and Wu(2024)]{ZW24}
S.~Zhou and X.~Wu.
\newblock Approximately {EFX} allocations for indivisible chores.
\newblock \emph{Artificial Intelligence}, 326\penalty0 (C), 2024.
\newblock ISSN 0004-3702.

\bibitem[Zhou et~al.(2024)Zhou, Wei, Li, and Li]{landscape24}
Y.~Zhou, T.~Wei, M.~Li, and B.~Li.
\newblock A complete landscape of {EFX} allocations on graphs: Goods, chores
  and mixed manna.
\newblock In \emph{Proceedings of the Thirty-Third International Joint
  Conference on Artificial Intelligence, {IJCAI-24}}, 2024.

\end{thebibliography}

\begin{appendix}
    \section{Finding EFX allocation on bipartite multi-graph with additive valuations: An Example of the Algorithm in Section \ref{sec:monotone}}\label{appendixA}

    In this section, we depict the execution of our algorithm for computing an $\efx$ allocation on bipartite multi-graphs on an input instance with additive valuations. We present an example with additive valuations for the ease of understanding. 

    Figure \ref{fig:example-1} depicts our input instance. Note that the instance is symmetric, and the numbers on edges specifies the value of that item for both endpoint agents. Note that, in this instance, $q = 2$ and the cut configurations are partitions into two bundles where the size of each bundle is at most one. As we run our algorithm, if we direct an edge towards one of the endpoint agents, it means that the edge is being allocated to that agent. Also, the set of envied vertices are colored blue.

    \begin{figure} [h]
    \centering
    \begin{tikzpicture}
    
    \Vertex[label = 1, x = 2, y = 2, color = white]{1}
    \Vertex[label = 2, x = 4, y = 2, color = white]{2}
    \Vertex[label = 3, x = 6, y = 2, color = white]{3}
    \Vertex[label = 4, x = 8, y = 2, color = white]{4}
    \Vertex[label = 5, x = 7, y = 4, color = white]{5}
    \Vertex[label = 6, x = 1, y = 4, color = white]{6}
    \Vertex[label = 7, x = 5, y = 0, color = white]{7}
    
    \Edge[label = $10$, bend=5](1)(5)
    \Edge[label = $10$, bend=10](2)(5)
    \Edge[label = $9$, bend=10](5)(2)
    \Edge[label = $8$, bend=0](5)(3)
    \Edge[label = $6$, bend=10](1)(6)
    \Edge[label = $5$, bend=-10](1)(6)
    \Edge[label = $6$, bend=10](2)(6)
    \Edge[label = $6$, bend=-10](2)(6)
    \Edge[label = $7$, bend=5](6)(3)
    \Edge[label = $6$, bend=10](1)(7)
    \Edge[label = $5$, bend=-10](1)(7)
    \Edge[label = $6$, bend=10](2)(7)
    \Edge[label = $6$, bend=-10](2)(7)
    \Edge[label = $7$, bend=0](7)(3)
    \Edge[label = $3$, bend=5](7)(4)
    \Edge[label = $4$, bend=-5](7)(4)
    \Edge[label = $6$, bend=10](5)(4)
    \Edge[label = $3$, bend=10](4)(5)
    
    \end{tikzpicture}
    
    \caption{An instance of a bipartite multi-graph with bi-partitions $S = \{1,2,3,4\}$ and $T = \{5,6,7\}$. We will depict the three main steps of our algorithm for finding an $\efx$ allocation using this example. The numbers on the edges depict the value of that edge for both endpoints.}
    \label{fig:example-1}
\end{figure}
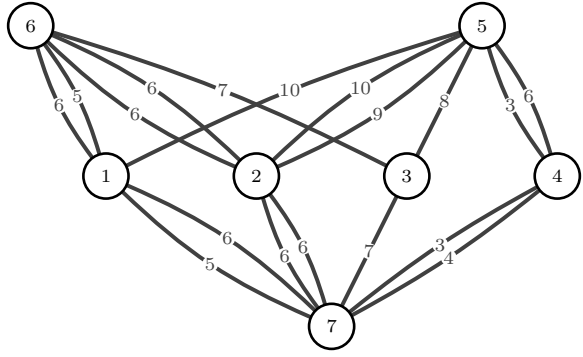

We begin by executing Algorithm \ref{alg:greedy} on our input instance to satisfy properties (1)-(3). The resulting allocation is depicted in Figure \ref{fig:example-2}. Next, we run Algorithm \ref{alg:prop-4} to satisfy property (4). The resulting allocation is depicted in Figure \ref{fig:example-3}. It is easy to verify that the said properties are indeed satisfied.

\begin{figure}[h]
    \centering
    \begin{tikzpicture}
    
    \Vertex[label = 1, x = 2, y = 2]{1}
    \Vertex[label = 2, x = 4, y = 2]{2}
    \Vertex[label = 3, x = 6, y = 2, color = white]{3}
    \Vertex[label = 4, x = 8, y = 2, color = white]{4}
    \Vertex[label = 5, x = 7, y = 4, color = white]{5}
    \Vertex[label = 6, x = 1, y = 4, color = white]{6}
    \Vertex[label = 7, x = 5, y = 0, color = white]{7}
    
    \Edge[label = $10$, bend=-5, Direct](5)(1)
    \Edge[label = $10$, bend=-10, Direct](5)(2)
    \Edge[label = $9$, bend=-10, Direct](2)(5)
    \Edge[label = $8$, bend=0, Direct](5)(3)
    \Edge[label = $6$, bend=10](1)(6)
    \Edge[label = $5$, bend=-10](1)(6)
    \Edge[label = $6$, bend=10](2)(6)
    \Edge[label = $6$, bend=-10](2)(6)
    \Edge[label = $7$, bend=-5, Direct](3)(6)
    \Edge[label = $6$, bend=10](1)(7)
    \Edge[label = $5$, bend=-10](1)(7)
    \Edge[label = $6$, bend=10](2)(7)
    \Edge[label = $6$, bend=-10](2)(7)
    \Edge[label = $7$, bend=0, Direct](3)(7)
    \Edge[label = $3$, bend=5](7)(4)
    \Edge[label = $4$, bend=-5](7)(4)
    \Edge[label = $6$, bend=10, Direct](5)(4)
    \Edge[label = $3$, bend=10](4)(5)
    
    \end{tikzpicture}
    
    \caption{Partial $\efx$ orientation obtained after executing Algorithm~\ref{alg:greedy} satisfies Properties (1)-(3). The vertices in blue are envied.}
    \label{fig:example-2}
\end{figure}

\begin{figure}[h]
    \centering
    \begin{tikzpicture}
    
    \Vertex[label = 1, x = 2, y = 2]{1}
    \Vertex[label = 2, x = 4, y = 2]{2}
    \Vertex[label = 3, x = 6, y = 2, color = white]{3}
    \Vertex[label = 4, x = 8, y = 2, color = white]{4}
    \Vertex[label = 5, x = 7, y = 4, color = white]{5}
    \Vertex[label = 6, x = 1, y = 4, color = white]{6}
    \Vertex[label = 7, x = 5, y = 0, color = white]{7}
    
    \Edge[label = $10$, bend=-5, Direct](5)(1)
    \Edge[label = $10$, bend=-10, Direct](5)(2)
    \Edge[label = $9$, bend=-10, Direct](2)(5)
    \Edge[label = $8$, bend=0, Direct](5)(3)
    \Edge[label = $6$, bend=10, Direct](1)(6)
    \Edge[label = $5$, bend=-10](1)(6)
    \Edge[label = $6$, bend=10, Direct](2)(6)
    \Edge[label = $6$, bend=-10](2)(6)
    \Edge[label = $7$, bend=-5, Direct](3)(6)
    \Edge[label = $6$, bend=10, Direct](1)(7)
    \Edge[label = $5$, bend=-10](1)(7)
    \Edge[label = $6$, bend=10, Direct](2)(7)
    \Edge[label = $6$, bend=-10](2)(7)
    \Edge[label = $7$, bend=0, Direct](3)(7)
    \Edge[label = $3$, bend=-5, Direct](4)(7)
    \Edge[label = $4$, bend=-5, Direct](7)(4)
    \Edge[label = $6$, bend=10, Direct](5)(4)
    \Edge[label = $3$, bend=10, Direct](4)(5)
    
    \end{tikzpicture}
    
    \caption{Partial $\efx$ orientation obtained after executing Algorithm~\ref{alg:prop-4} satisfies properties (1)-(4). Note that, Case (1) is executed for the pair $(4,5)$ of vertices, Case (2) for the pair $(7,4)$ of vertices, and Case (3) for the pairs $(6,2), (6,1), (7,2), (7,1),$ and $(7,3)$ of vertices. The vertices in blue are envied. It is easy to check the correctness of Claim \ref{claim:non-envied-I} for non-envied vertices.}
    \label{fig:example-3}
\end{figure}

Observe that, property (5) is automatically \emph{satisfied} in the output (partial) orientation of Algorithm~\ref{alg:prop-4} , and hence, we do not need to run Algorithm \ref{alg:non-envied-unalloc}. In fact, for any input instance with additive valuations, if properties (1)-(4) are satisfied, then property (5) is also satisfied. Note that this does not necessarily hold when the valuations are monotone. We formalize this claim as follows:

\begin{restatable}{claim}{claimnonenviedI}\label{claim:non-envied-I}
        For any fair division instance on a bipartite multi-graph with additive valuations, after satisfying properties (1)-(4) (using Algorithms~\ref{alg:greedy} and \ref{alg:prop-4}, executed in that order), we have $v_i(U_i(X)) \le v_i(X_i)$ for every non-envied vertex $i \in [n]$, i.e., property (5) is also satisfied.
\end{restatable}
\begin{proof}
    Consider a non-envied vertex $i \in [n]$ in the output (partial) orientation $X$ (obtained after running Algorithms~\ref{alg:greedy} and \ref{alg:prop-4}). By property (4), $A_{i, j}(X) = \emptyset$ for all $j \neq i$. We define the set $R = \{k \in [n] :  A_{k, i}(X) \neq \emptyset\}$. Observe that by Claim~\ref{claim:unallocated_edges}, any $k \in R$ must be envied and $U_i(X) = \bigcup\limits_{k \in R} A_{k, i}(X)$. We assume $R \neq \emptyset$; otherwise, the claim holds trivially for $i$. We now consider an arbitrary agent $k \in R$. Again, Claim \ref{claim:unallocated_edges} implies that the set $E(k, i) \setminus A_{k, i}(X)$ is allocated to $i$. Note that agent $i$ has received this set either in Algorithm \ref{alg:greedy} or in the third case of Algorithm \ref{alg:prop-4}. Since, in both cases agent $i$ had the possibility of choosing between $A_{k, i}(X)$ and $E(k ,i) \setminus A_{k, i}(X)$, we have $v_i(E(k, i) \setminus A_{k, i}(X)) \ge v_i(A_{k, i}(X))$. Thus, we can write,
    \begin{align*}
        v_i(U_i(X)) &= v_i(\cup_{k \in R} A_{k,i}(X)) \\&\le v_i(\bigcup\limits_{k \in R} (E(k, i) \setminus A_{k, i}(X))) \\&\le v_i(X_i),
    \end{align*}

    where the last inequality holds due the fact that $E(k, i) \setminus A_{k, i}(X) \subseteq X_i$ for all $k \in R$ and these sets do not intersect with each other.
\end{proof}

Since property (5) is automatically satisfied for additive valuations, we skip Algorithm~\ref{alg:non-envied-unalloc} and directly run Algorithm \ref{alg:safe-set-monotone}, which satisfies property (6) while maintaining the first five properties. Figure \ref{fig:example-4} depicts the output allocation at this point. 

\begin{figure}[h]
    \centering
    \begin{tikzpicture}
    
    \Vertex[label = 1, x = 2, y = 2]{1}
    \Vertex[label = 2, x = 4, y = 2, color = white]{2}
    \Vertex[label = 3, x = 6, y = 2, color = white]{3}
    \Vertex[label = 4, x = 8, y = 2, color = white]{4}
    \Vertex[label = 5, x = 7, y = 4, color = white]{5}
    \Vertex[label = 6, x = 1, y = 4, color = white]{6}
    \Vertex[label = 7, x = 5, y = 0, color = white]{7}
    
    \Edge[label = $10$, bend=-5, Direct](5)(1)
    \Edge[label = $10$, bend=10, Direct](2)(5)
    \Edge[label = $9$, bend=10, Direct](5)(2)
    \Edge[label = $8$, bend=0, Direct](5)(3)
    \Edge[label = $6$, bend=10, Direct](1)(6)
    \Edge[label = $5$, bend=-10](1)(6)
    \Edge[label = $6$, bend=10, Direct](2)(6)
    \Edge[label = $6$, bend=10, Direct](6)(2)
    \Edge[label = $7$, bend=-5, Direct](3)(6)
    \Edge[label = $6$, bend=10, Direct](1)(7)
    \Edge[label = $5$, bend=-10](1)(7)
    \Edge[label = $6$, bend=10, Direct](2)(7)
    \Edge[label = $6$, bend=10, Direct](7)(2)
    \Edge[label = $7$, bend=0, Direct](3)(7)
    \Edge[label = $3$, bend=-5, Direct](4)(7)
    \Edge[label = $4$, bend=-5, Direct](7)(4)
    \Edge[label = $6$, bend=10, Direct](5)(4)
    \Edge[label = $3$, bend=10, Direct](4)(5)
    
    \end{tikzpicture}
    
    \caption{Partial $\efx$ orientation $X$ obtained after executing Algorithm~\ref{alg:safe-set-monotone} satisfies properties (1)-(6). Agent $5$ envies agent $2$ in $X$ and $5 \notin S_2(X)$. Hence, agent $2$ gives the item of value ten to agent $5$ and receives the item with value nine from agent $5$, and then receives $U_2(X)$. Agent $5$ envies agent $1$ in $X$ and $5 \in S_1(X)$. In the figure, the only unallocated edges are between agents $6$ and $1$, and $7$ and $1$. In the last step, since agent $5$ envies agent $1$, we can allocate both the unallocated edges to agent $5$, achieving a complete $\efx$ allocation.}
    \label{fig:example-4}
\end{figure}
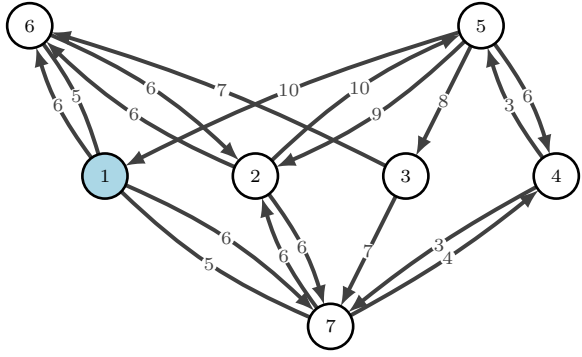

Note that, in the output partial $\efx$ orientation $X$ of Algorithm~\ref{alg:safe-set-monotone}, it is only agent $5$ who envies agent $1$. Therefore, we allocate all the remaining edges to agent $5$, yielding a complete $\efx$ allocation.

\end{appendix}
\end{document}